\documentclass[12pt]{article}
\usepackage{bbm,fullpage}
\usepackage{bm}
\usepackage{amsmath,amssymb,amsthm,amsmath}
\usepackage{algorithm, pseudocode}
\usepackage{mathrsfs}
\usepackage{enumitem}

\usepackage[titles]{tocloft}
\usepackage{yfonts}
\usepackage{xcolor}
\usepackage[pdftex,                %
    bookmarks         = true,
    bookmarksnumbered = true,
    pdfpagemode       = None,
    pdfstartview      = FitH,
    pdfpagelayout     = SinglePage,
    colorlinks        = true,
   citecolor         =blue,
    urlcolor          = magenta,
    pdfborder         = {0 0 0}
    ]{hyperref}
       
\usepackage{bbm}
\usepackage{amsmath}
\usepackage{easybmat}
\usepackage{multirow,bigdelim}

\usepackage[indexonlyfirst,ucmark,toc]{glossaries}

\definecolor{ao(english)}{rgb}{0.0, 0.5, 0.0}

\def\softO{\ensuremath{{O}{\,\tilde{ }\,}}}

\def\supp{\ensuremath{\text{supp}}}
\def\conv{\ensuremath{\text{conv}}}

\def\calA{\ensuremath{\mathcal{A}}}
\def\calB{\ensuremath{\mathcal{B}}}
\def\calC{\ensuremath{\mathcal{C}}}
\def\calD{\ensuremath{\mathcal{D}}}

\def\calM{\ensuremath{\mathcal{M}}}

\def\N {\ensuremath{\mathbb{N}}}
\def\Q {\ensuremath{\mathbb{Q}}}
\def\R {\ensuremath{\mathbb{R}}}
\def\Z {\ensuremath{\mathbb{Z}}}

\def\KKbar {\ensuremath{\overline{\mathbf{K}}}}
\def\KK {\ensuremath{{\mathbf{K}}}}

\def\OpenBKK{\ensuremath{\mathscr{O}_{\rm BKK}}}
\def\OpenO{\ensuremath{\mathscr{O}}}

\def\Openinit{\ensuremath{\Omega_2}}
\def\Openstart{\ensuremath{\Omega_1}}

\def\Openbounded{\ensuremath{\Omega_3}}
\def\Open{\ensuremath{\Omega}}
\DeclareBoldMathCommand{\bw}{w}

\def\KKbar {\ensuremath{\overline{\mathbf{K}}}}
\def\rank{\ensuremath{{\rm rank}}}
\def\wdeg{\ensuremath{{\rm wdeg}}}

\def\init{\ensuremath{{\rm init}}}

\DeclareBoldMathCommand{\e}{e}
\DeclareBoldMathCommand{\qq}{q}
\DeclareBoldMathCommand{\f}{f}
\DeclareBoldMathCommand{\p}{p}
\DeclareBoldMathCommand{\g}{g}
\DeclareBoldMathCommand{\k}{k}
\DeclareBoldMathCommand{\r}{r}
\DeclareBoldMathCommand{\h}{h}
\DeclareBoldMathCommand{\u}{u}
\DeclareBoldMathCommand{\x}{x}
\DeclareBoldMathCommand{\w}{w}
\DeclareBoldMathCommand{\j}{i}
\DeclareBoldMathCommand{\i}{i}
\DeclareBoldMathCommand{\bell}{\ell}

\DeclareBoldMathCommand{\bX}{x}
\DeclareBoldMathCommand{\bY}{Y}
\DeclareBoldMathCommand{\bZ}{Z}
\DeclareBoldMathCommand{\bZ}{Z}

\DeclareBoldMathCommand{\bcalA}{\mathcal{A}}
\DeclareBoldMathCommand{\bcalB}{\mathcal{B}}
\DeclareBoldMathCommand{\bcalC}{\mathcal{C}}

\DeclareBoldMathCommand{\balpha}{\alpha}
\DeclareBoldMathCommand{\bbeta}{\beta}
\DeclareBoldMathCommand{\bkappa}{\kappa}
\DeclareBoldMathCommand{\blambda}{\lambda}

\def\vol{\ensuremath{\mathrm{Vol}}}

\def\F {\ensuremath{\bm{F}}}

\def\M {\ensuremath{\bm M}}

\def\frkc{\ensuremath{\mathfrak{c}}}
\def\frkd{\ensuremath{\mathfrak{d}}}
\def\frke{\ensuremath{\mathfrak{e}}}
\def\frkf{\ensuremath{\mathfrak{f}}}

\def\frkm{\ensuremath{\mathfrak{m}}}

\def\frkr{\ensuremath{\mathfrak{r}}}

\def\frkA{\ensuremath{\mathfrak{A}}}
\def\frkM{\ensuremath{\mathfrak{M}}}

\def\frkB{\ensuremath{\mathfrak{B}}}

\def\frkt{\ensuremath{\mathfrak{t}}}

\DeclareBoldMathCommand{\bA}{A}
\DeclareBoldMathCommand{\bB}{B}
\DeclareBoldMathCommand{\bC}{C}
\DeclareBoldMathCommand{\bW}{W}
\DeclareBoldMathCommand{\bH}{H}

\def\scrR{\ensuremath{\mathscr{R}}}

\def\vol{\ensuremath{\mathrm{vol}}}

\def\wdeg{\ensuremath{\mathrm{wdeg}}}

\def\mycomment#1{}

\title{Homotopy techniques for solving sparse column support determinantal polynomial systems}
\author{George Labahn \thanks{David R. Cheriton School of Computer Science,
        University of Waterloo, Waterloo ON, Canada N2L 3G1, emails:{\text{\{glabahn, eschost, txvu\}@uwaterloo.ca}}}, 
        Mohab Safey El Din \thanks{Sorbonne Universit\'e, CNRS, Laboratoire
          d'Informatique de Paris 6 (LIP6, UMR7606), \'Equipe POLSYS, 4 place Jussieu, F-75252, Paris Cedex 05, France, email:{\text{Mohab.Safey@lip6.fr}}}, 
        \'Eric Schost$^*,$  Thi Xuan Vu$^{\dagger *}$
}

\date{}

\def\elim{\mathfrak{w}}

\newtheorem{definition}{Definition}
\numberwithin{definition}{section}
\newtheorem{theorem}[definition]{Theorem}

\newtheorem{proposition}[definition]{Proposition}
\newtheorem{lemma}[definition]{Lemma}

\begin{document}

\maketitle 

\begin{abstract}
Let $\KK$ be a field of characteristic zero with $\KKbar$ its
algebraic closure. Given a sequence of polynomials $\g = (g_1, \ldots,
g_s) \in \KK[x_1, \ldots , x_n]^s$ and a polynomial matrix $\F =
[f_{i,j}] \in \KK[x_1, \ldots, x_n]^{p \times q}$, with $p \leq q$, we
are interested in determining the isolated points of $V_p(\F,\g)$, the
algebraic set of points in $\KKbar$ at which all polynomials in $\g$
and all $p$-minors of $\F$ vanish, under  the assumption
$n = q - p + s + 1$. Such polynomial systems arise in a
variety of applications including for example polynomial optimization
and computational geometry.

We design a randomized sparse homotopy algorithm for computing the isolated
points in $V_p(\F,\g)$ which takes advantage of the determinantal structure of
the system defining $V_p(\F,\g)$. Its complexity is polynomial in the maximum
number of isolated solutions to such systems sharing the same sparsity pattern
and in some combinatorial quantities attached to the structure of such systems.
It is the first algorithm which takes advantage both on the determinantal
structure and sparsity of input polynomials.

We also derive complexity bounds for the particular but important case where
$\g$ and the columns of $\F$ satisfy weighted degree constraints. Such systems
arise naturally in the  computation of critical points of maps restricted to algebraic 
sets when both are invariant by the action of the symmetric group.
\end{abstract}

\section{Introduction}

Let $\g= (g_1, \dots, g_s)$ be a sequence of polynomials in $\KK[x_1,
  \dots, x_n]^s$, and let $\F = [f_{i, j}]$ be a polynomial matrix in
$\KK[x_1, \dots, x_n]^{p \times q}$, where $\KK$ is a field of
characteristic zero with algebraic closure $\KKbar$.
Assuming  $p \leq q$, we are interested in describing the set 
\begin{equation}\label{eq:one}
  V_p(\F, \g) = \{\x \in \KKbar^n \, | \, \rank(\F(\x)) < p \, \mbox{ and } \, 
  g_1(\x) = \cdots = g_s(\x) = 0\}. 
\end{equation}
If for any positive integer $r$ we let $M_r(\F)$ be the set of all
$r$-minors of $\F$ then our set of points is given by
$$
V(\langle M_p(\F) \rangle + \langle g_1, \dots, g_s\rangle).
$$

As an example, when $\F$ denotes the Jacobian of $(g_1, \ldots, g_s,
\phi)$ with respect to the variables $x_1, \ldots, x_n$, for some $\phi \in
\KK[x_1, \dots, x_n]$, then $V_{s+1}(\F, \g)$ is the set of critical
points of $\phi$ over the algebraic set $V(\g)$, assuming $\g$ is a
reduced regular sequence and $V(\g)$ is smooth. The problem of
computing such points appear in many areas such as polynomial
optimization and real algebraic geometry. Note that in this example we
have $n = q-p+s+1$ (since $\F$ has dimensions $p=s+1$ and $q=n$); we
will assume that this holds throughout this paper.

We wish to describe the {\em isolated} zeros of our algebraic set $V_p(\F,
\g)$ when all entries of $\F$ and $\g$ are {\em sparse
  polynomials}. We also want to take advantage of the special
determinantal structure of our algebraic set to obtain complexity
results which are polynomial in the {\em generic} number of solutions
in $\KKbar^n$ of such systems (this is the number of solutions
obtained when the coefficients of terms appearing in the entries of
$\F, \g$ are algebraically independent indeterminates) and some
combinatorial data attached to the monomial structure of the entries.

In order to achieve this, we make use of the technique of {\it symbolic homotopy
  continuation} and show how it can be used to obtain a solver with such a good
complexity. Homotopy continuation has become a foundational tool for numerical
algorithms while the use of symbolic homotopy continuation algorithms is more
recent. Such algorithms first appeared in~\cite{Bomp04, Heintz00} without any
structure on the system. Later symbolic homotopies were used in square sparse
systems~\cite{Jer09, HeJeSa10, Maria13, HeJeSa14} and multi-homogeneous
systems~\cite{SaSc18, HeJeSaSo02, HSSV18}.
 
Homotopy continuation involves defining a deformation between our
system defining $V_p(\F,\g)$ and a second system defining $V_p({\M}, {\r})$ which is
similar but whose solutions are easy to describe. Formally, we let $t$
be a new variable and construct a matrix
 \begin{equation}\label{eqdef:V}
   {\bm V} =  (1-t) \cdot \M + t \cdot \F \in \KK[t, x_1, \ldots ,
   x_n]^{p\times q} 
 \end{equation}
  which connects a {\it start matrix} $\M \in \KK[x_1, \ldots ,
  x_n]^{p\times q}$ to our target matrix $\F$, together with
  polynomials $\u  = (u_1, \dots, u_s)$  of the form
 \begin{equation}\label{eqdef:u}
   \u =  (1-t) \cdot \r + t \cdot \g \in \KK[t, x_1, \ldots ,
   x_n]^s, 
 \end{equation}
 that connects a starting polynomial system $\r$ to our target system
 $\g$.  Such a homotopy allows us to define a {\em homotopy curve},
 steering the solutions of the start system to the isolated solutions
 to our input system (we do not assume that our input system has 
 finitely many solutions).

 We will use a data-structure known as {\em
   zero-dimensional parametrization} to represent finite algebraic
 sets.  If $V$ is such a set, defined by polynomials over $\KK$, a
 zero-dimensional parametrization $\scrR = ((\elim, v_1, \dots, v_n),
 \Lambda)$ of $V$ consists of
 \begin{itemize}
 \item[(i)] a square-free polynomial $\elim$ in $\KK[y]$, where $y$ is a
   new indeterminate,
 \item[(ii)] polynomials $(v_1, \dots, v_n)$ in $\KK[y]$ with each
   $\deg(v_i) < \deg(\elim)$ and satisfying
   \begin{center}
     $V = \{(v_1(\tau), \dots, v_n(\tau)) \in \KKbar^{n} \, | \, \elim(\tau) =
     0\}$,\end{center}
 \item[(iii)] a linear form $\Lambda = \lambda_1\, x_1 + \cdots +
   \lambda_n \, x_n$ with coefficients in $\KK$, such that $\lambda_1
   v_1 + \cdots + \lambda_n v_n~=y$ (so the roots of $\elim$ are the
   values taken by $\Lambda$ on $V$).
 \end{itemize}
When this holds, we write $V=Z(\scrR)$. This representation was
introduced in early work of Kronecker and
Macaulay~\cite{Kronecker82,Macaulay16} and has been widely used as a
data structure in computer algebra, see for instance~\cite{GiMo89,
  ABRW96, GiHeMoMoPa98,GiHeMoPa95, Rouillier99,GiLeSa01}.

Then, given a zero-dimensional parametrization $\scrR_0$ of $V_p(\M,
\r)$, we will apply the algorithm in~\cite{HSSV18} to the system
$(M_p({\bm V}), \u)$ to lift $\scrR_0$ to a zero-dimensional
parametrization $\scrR_1$ of the isolated zeros of $V_p(\F, \g)$.
At a high level the strategy for using homotopy methods to determine
isolated zeros is relatively simple to describe, but also difficult to
realize. The start system should have at least the same number of
solutions as the target system and should be `easy' to solve. Also, we
want a {\em sparse} homotopy algorithm, that is, we also wish to have
a complexity which depends on the support of the polynomials appearing
in our target system. 

The main contribution in this paper is to provide the needed
ingredients for a sparse homotopy algorithm for our determinantal
systems which makes use of the {\em column support} of $\F$. We
determine a family of possible start systems, and we show that a
generic member of this family allows us to carry out the procedure
successfully; we also show how to compute the solutions of this start
system. Our runtime is polynomial in the degree of the start system
and the degree of the homotopy curve, both depending on certain mixed
volumes related to the polynomials $\g$ and the columns of $\F$, see
Theorem~\ref{thm:sparsedeterminantal}.  As far as we are aware, this
is the first homotopy algorithm which simultaneously exploits both
determinantal structure and sparsity.

The tools used to create our sparse column support homotopy also allow
us to build a column homotopy algorithm for determinantal systems for
weighted degree polynomials. These are important when all our input
polynomials (including those in the input matrix) are invariant under
the action of the group of permutations on $n$ letters. In that case,
one can perform an algebraic change of coordinates to express all
entries with respect to elementary symmetric functions which are
naturally weighted (the $k$-th elementary symmetric function then has
weighted degree $k$). We show that one obtains a speed-up which is
polynomial in the product of the weights, see Theorem~\ref{thm:homotopy}.

This is not the first time that determinantal structures have been exploited to
speed-up polynomial system solvers. Previous work includes, for example,
\cite{HSSV18}, which is also based on homotopy techniques: we borrow some
results and techniques from that reference, but our discussion of the ``sparse''
aspects is new. Note also that one can encode rank deficiencies in a polynomial
matrix using extra variables (sometimes called Lagrange multipliers in the
context of polynomial optimization) to encode that the kernel of the considered
matrix is non-trivial. This would lead to Lagrange systems with a sparse
structure, which could be solved using homotopy techniques from~\cite{Jer09,
  HeJeSa10, Maria13, HeJeSa14}. However, this technique does not work when
isolated solutions to our determinantal system lead to rank deficiencies higher
than one: such isolated points of our determinantal system do not correspond to
isolated points of the Lagrange system. Still, we will see that such systems
play an important role to prove intermediate results needed to achieve our
results.

The use of geometric resolution algorithms is investigated in the series of
works \cite{BGHLM, BGHP, BGHP05, SP16} (and references therein). In this latter
setting, relating the complexity parameters (which are mainly geometric degrees
of some algebraic sets defined by the input) with the sparsity of these inputs
is still a non-trivial problem. Determinantal systems in the context of
Gr\"obner bases are also considered in \cite{FSS10, FSP12, Sp14}. Again, this
series of works do not take into account the sparsity of the entries.

The structure of the paper is as follows. Section~\ref{sec:pre} gives
some of the preliminary background on sparse polynomials; it is
followed by Section~\ref{sec:homotopy} which introduces the template
of a homotopy algorithm and states properties that will guarantee it
succeeds; at this stage, we do not specify how to choose the start
system. In Section~\ref{sec:colum_homotopy}, we introduce a family of
start systems and prove that a generic member of this family satisfies
the properties needed for our symbolic homotopy algorithm. The cost of
our algorithm is analyzed in Section~\ref{sec:complexity}, first
in the general case of sparse polynomials, then in the important case
of weighted domains. An example illustrating the steps of our homotopy
algorithm is given in Section \ref{sec:ex}. The paper ends with a
conclusion and topics for future research.

\section{Preliminaries}
\label{sec:pre}

\paragraph{Sparse polynomials.}
Consider a set $\bX = (x_1, \dots, x_n)$ of
indeterminates. Polynomials in $\bX$ are represented in the form of
finite sums $f = \sum_{\balpha =(\alpha_1,\dots,\alpha_n) \in \calA}
c_{\balpha}x_1^{\alpha_1} \cdots x_n^{\alpha_n}$, with $\calA$ being a finite
subset of $\N^n$, the set
$\{\balpha \in \N^n \, : \, c_{\balpha} \ne 0\}\subset \calA$ being the
\textit{support} $\supp(f)$ of $f$.  The \textit{Newton polytope} of
$f$, denoted by $\conv(f)$, is the convex hull of the support of $f$
in $\R^n$.

We will often work in the following setup. Consider $\ell$ finite sets
$\calA_1, \dots, \calA_\ell$ in $\N^n$, with $k_i$ denoting the
cardinality of $\calA_i$ for all $i$. For each $i$, we let $\calM_i = (m_{i,
  1}, \dots, m_{i, k_i})$ be the corresponding set of monomials
in $x_1, \ldots, x_n$. This allows us to define the ``generic
polynomials'' $\frkf_1, \ldots, \frkf_\ell$ supported on $\calA_1,
\dots, \calA_\ell$ by
\[
\frkf_i = \sum_{j=1}^{k_i} \, \frkc_{i, j} m_{i, j} \in
\KK[\frak C][x_1, \ldots, x_n], 
\]
where $\frak C=(\frkc_{i, j})_{1 \le i \le \ell, 1 \le j \le k_i}$
are new indeterminates. The total number of indeterminates $\frak C$ is
$N = \sum_{i=1}^\ell  k_i$. 

Identifying $\KKbar{}^{N}$ with $\KKbar{}^{k_1} \times \cdots
\times \KKbar{}^{k_\ell}$, we can view any element $\rho \in
\KKbar{}^{N}$ as a vector of coefficients, first for $\frkf_1$, then
for $\frkf_2$, etc. Then, for such a $\rho$, we will denote by
$\Theta_\rho$ the mapping
\begin{eqnarray*}
\KK[\frak C][x_1,\dots,x_n]&\to&\KKbar[x_1,\dots,x_n]\\
\sum_{\balpha \in \N^n}  \frkc_{i,j} x_1^{\alpha_1} \cdots x_n^{\alpha_n} &\mapsto& 
\sum_{\balpha \in \N^n}  \rho_{i,j} x_1^{\alpha_1} \cdots x_n^{\alpha_n}.
\end{eqnarray*}
The notation carries over to vectors or matrices of polynomials.  In
this paragraph, we discuss some properties of the zeros of systems
$\Theta_\rho(\frkf_1,\dots,\frkf_\ell)$.

For the first proposition, $\ell$ is arbitrary, but we impose a
restriction on the sets $\calA_i$.
\begin{proposition} \label{prop:generic_single} 
  Suppose that for $i=1,\dots,\ell$, $\calA_i$ contains the origin
  $\bm 0 \in \N^n$.  Then there exists a non-empty Zariski open set
  $\OpenO \subset \KKbar{}^N$ such that for $\rho \in \OpenO$, we
  have the following:
  \begin{itemize}
  \item[(i)] if $\ell \le n$, $\Theta_\rho(\frkf_1, \dots,
    \frkf_\ell)$ generates a radical ideal, whose zero-set in
    $\KKbar{}^n$ is either empty or smooth and $(n-\ell)$-equidimensional;
\smallskip
  \item[(ii)] if $\ell > n$, the zero-set of $\Theta_\rho(\frkf_1, \dots,
    \frkf_\ell)$ in $\KKbar{}^n$ is empty.
  \end{itemize}
\end{proposition}
\begin{proof}
  Without loss of generality, assume that $m_{i, k_i} = 1$ holds since we assume
  that $\calA_i$ contains the origin $\bm 0 \in \N^n$ for all $1\leq i\leq
  \ell$. Consider the mapping
  \[
  \Phi: (\x, \rho)\in \KKbar{}^n \times \KKbar{}^{N}\mapsto
  \Theta_{\rho}({\frak f}_1, \ldots, 
  {\frak f}_\ell)(\x). 
  \]
  We first claim that $\mathbf{0}$ is a regular value of $\Phi$, that
  is, the Jacobian matrix of this sequence of polynomials has full
  rank at all points $(\x, \rho)$ of its zero-set.  Indeed, since
  $m_{i, k_i} = 1,$ the columns corresponding to partial
  derivatives with respect to $\frak{C}$ contain an
  $\ell\times\ell$ identity matrix.
  
  As a result, by Thom's weak transversality theorem (see the
  algebraic version in e.g.~\cite{SafeySchost17}), there exists a
  non-empty Zariski open set $\OpenO \subset \KKbar{}^{N}$ such
  that for $\rho$ in $\OpenO$, $\bf 0$ is a regular value of the
induced mapping
  \[
  \Phi_{\rho} : \x \in \KKbar{}^n \mapsto \Theta_{\rho}({\frak f}_1,
  \ldots, {\frak f}_\ell)(\x).
  \]
  In other words, the Jacobian matrix of $\Theta_{\rho}({\frak f}_1,
  \ldots, {\frak f}_\ell)$ has rank $\ell$ at any zero $\x\in
  \KKbar{}^n$ of $\Theta_{\rho}({\frak f}_1, \ldots, {\frak
    f}_\ell)$. For $\ell \le n$, by the Jacobian criterion
  \cite[Theorem 16.19]{Eisenbud95}, the ideal $\langle
  \Theta_{\rho}({\frak f}_1, \ldots, {\frak f}_\ell) \rangle$ is
  therefore radical, and its zero-set is either empty or smooth and
  $(n-\ell)$-equidimensional. For $\ell > n$, this means that this set
  is empty (since the matrix above has $n$ columns, it cannot
  have rank $\ell$).
\end{proof}

For the next properties, we take $\ell=n$. In what follows,
$\calC_1,\dots,\calC_n$ are the convex hulls of $\calA_1,\dots,\calA_n$,
respectively, with the Euclidean volume of $\calC_i$ in $\R^n$ being denoted by
${\rm vol}_{\R^n}(\calC_i)$. Consider the function
\[
\varphi \, : \, (\lambda_1, \dots, \lambda_n) \mapsto
\vol_{\R^n}(\lambda_1\calC_1 + \cdots + \lambda_n\calC_n),
\]  
where
$$\lambda_1\calC_1 + \cdots + \lambda_n\calC_n = \{\x \in \R^n \, : \,
\x = \sum_{i=1}^n \lambda_i\, x_i \, {\rm with} \, x_i \in \calC_i
\}$$ is the Minkowski sum of polytopes.  The function $\varphi$ is a
homogeneous polynomial function of degree $n$ in $\lambda_i$ (see
e.g.\ \cite[Proposition~4.9]{Cox2006}). The \textit{mixed volume}
${\sf MV}(\calC_1, \dots, \calC_n)$ is then defined as the coefficient
of the monomial $\lambda_1 \cdots \lambda_n$ in $\varphi$. Then, the
Bernstein-Khovanskii-Kushnirenko (BKK) theorem~\cite{Bernstein75}
gives a bound on the number of isolated zeros of
$\Theta_\rho(\frkf_1,\dots,\frkf_n)$ in the torus in terms of this
quantity (note that here, we do not assume that the supports $\calA_i$
contain the origin).

\begin{proposition}\label{theorem:bers75-torus} 
  For any $\rho$ in $\KKbar{}^N$, the number of isolated zeros of
  $\Theta_\rho(\frkf_1,\dots,\frkf_n)$ in
  $(\KKbar-\{0\})^n$ is at most ${\sf MV}(\calC_1, \dots, \calC_n)$. 
  Furthermore, there exists a non-empty Zariski-open set $\OpenBKK
  \subset \KKbar{}^N$ such that the bound is tight for $\rho$ in
  $\OpenBKK$.
\end{proposition}

A first application of Proposition~\ref{prop:generic_single} is the
following refinement of this statement (which of course requires the
assumptions of Proposition~\ref{prop:generic_single} to hold). Again,
we take $\ell=n$.
\begin{proposition}\label{theorem:bers75}
  Suppose that for $i=1,\dots,n$, $\calA_i$ contains the origin $\bm 0
  \in \N^n$. Then, there exists a non-empty Zariski-open set
  $\OpenBKK' \subset \KKbar{}^N$ such that for $\rho$ in $\OpenBKK'$,
  $\Theta_\rho(\frkf_1,\dots,\frkf_n)$ has ${\sf MV}(\calC_1, \dots,
  \calC_n)$ solutions in $\KKbar^n$.
\end{proposition}
\begin{proof}
  Consider a subset $\j=\{i_1,\dots,i_m\}$ of $\{1,\dots,n\}$, with $1
  \le m \le n$, and let $(\frkf_{\j,1},\dots,\frkf_{\j,n})$ be the
  polynomials $(\frkf_1,\dots,\frkf_n)$ where the coordinates
  $x_{i_1},\dots,x_{i_m}$ have been set to zero; they depend on a
  certain number $N_\j \le N$ of indeterminate coefficients $\rho_\j$.

  This is thus a system of $n$ equations in $n-m < n$ unknowns, and
  the support of each of these equations still contains the origin.
  Proposition~\ref{prop:generic_single} then implies that there exists
  a non-empty Zariski-open $\omega_\j \subset \KKbar{}^{N_\j}$ such
  that for $\rho_\j$ in $\omega_\j$,
  $\Theta_{\rho_\j}(\frkf_{\j,1},\dots,\frkf_{\j,n})$ has no solution
  in $\KKbar{}^{n-m}$. Let then $\Omega_\j$ be the preimage of
  $\omega_\j$ in $\KKbar{}^N$ (under the canonical projection), and
  define $\Omega$ as the intersection of all $\Omega_\j$, for
  $\j=\{i_1,\dots,i_m\}$ a subset of $\{1,\dots,n\}$. For $\rho$ in
  $\Omega$, all coordinates of all solutions of
  $\Theta_{\rho}(\frkf_1,\dots,\frkf_n)$ are non-zero. To conclude, we
  define $\OpenBKK'$ as the intersection of $\OpenBKK$ (from
  Proposition~\ref{theorem:bers75-torus}) and $\Omega$.
\end{proof}
  
\paragraph{Initial forms.}
Let $\e = (e_1, \dots, e_n)$ be non-zero in $\Q^n$ and consider a
 polynomial
$$p = \sum_{\balpha=(\alpha_1,\dots,\alpha_n) \in S} c_{\balpha} \, x_1^{\alpha_1} \cdots
x_n^{\alpha_n}$$ with support $S=\supp(p)$. The field of definition may be
our field $\KK$, or, as will also happen below, a rational function field.
Define
\[
m(\e, p) = \min(\langle \e, \balpha
  \rangle \, | \, \balpha \in S) \quad
{\rm and}
\quad  S_{\e, p} = \{\balpha \in S
  \, | \, \langle \e, \balpha \rangle = m(\e, p)\},
\]
where $\langle \ , \ \rangle$ is the usual dot-product in $\R^n$.
Thus, $S_{\e, p}$ is the intersection of $S$ with its ``support
hyperplane'' in the direction $\e$. The {\em initial form}
of $p$ with respect to $\e$  is defined as
\[
\init_\e(p) = \sum_{\balpha=(\alpha_1,\dots,\alpha_n) \in S_{\e, p}}
\,c_\balpha\, x_1^{\alpha_1} \cdots x_n^{\alpha_n}. 
\] In other words, $\init_\e(p)$ is the sum over all terms
$c_\balpha \, x_1^{\alpha_1} \cdots x_n^{\alpha_n}$ for which the dot-product
$\langle \e, \balpha \rangle $ is minimized.  For a vector $\p = (p_1,
\dots, p_n)$ of  polynomials, we let $$\init_\e(\p) =
(\init_\e(p_1), \dots, \init_\e(p_n)).$$ Even though there is an
infinite number of possible directions $\e$, the number of polynomial
systems $\{\init_\e(\p) \, | \, \e \text{~non-zero~in~} \Q^n\}$
obtained in this manner is finite, since the support of each $p_i$ has
finitely many support hyperplanes.

\section{Determinantal homotopy}
\label{sec:homotopy}
\label{subsec:hom_properties}
  
In this section, we review a few useful properties of homotopy
continuation methods for determinantal ideals. As input, we are given
$\g=(g_1,\dots,g_s)$ and $\F$ in $\KK[x_1,\dots,x_n]^{p \times q}$,
and we assume $n=q-p+s+1$. Let $t$ be a new variable and construct a
matrix
\[
{\bf V} =  (1-t) \cdot \M + t \cdot \F\in \KK[t, x_1, \ldots , x_n]^{p\times q}
\]
which connects a {\it start matrix} $\M \in \KK[x_1, \ldots ,
x_n]^{p\times q}$ to our target matrix $\F$, together with polynomials
$\u = (u_1, \dots, u_s)$ of the form
\[
\u =  (1-t) \cdot \r + t \cdot \g \in \KK[t, x_1, \ldots , x_n]^s,
\] 
which connect a starting polynomial system $\r$ to our target system
$\g$. Then, ${\bf V}$ and $\u$ define a deformation which allows us to
connect the solutions of the start system $V_p(\M,\r)$ to the
isolated solutions of our system $V_p(\F,\g)$. 

Algorithms for symbolic homotopy continuation require several
ingredients. We need a start system that can be solved efficiently and
has the ``right'' number of solutions, a description of the 
solutions of this start system, and a bound $\varrho$ that determines
the number of steps we perform. 

Proposition~\ref{pro:start_mat} below makes these requirements more
precise; it is a minor modification
of~\cite[Propositions~13~and~24]{HSSV18}.  To state it, it will be
convenient to describe our homotopy process using only vectors of
polynomials. To this end, we fix an ordering $\succ$ on the $p$-minors
of $p\times q$ matrices and set $m = s+{{q} \choose p}$.
Consider the system of equations
$$\bB= (u_1, \dots,u_s, b_{s+1},\dots,b_m) \in \KK[t, x_1, \dots,
  x_n]^m,$$ where $u_1,\dots,u_s$ are as defined above, and where the
polynomials $(b_{s+1}, \dots, b_m)$ are the $p$-minors of ${\bf V}$,
following the ordering $\succ$.  For $\tau \in \KKbar$, we write
$\bB_{t = \tau}$ for the polynomials in $\KKbar[x_1,\dots,x_n]$
obtained by the evaluation $t \mapsto \tau$ in $\bB$. In particular,
$\bB_{t = 0}$ is the set of equations in our start system, and
$\bB_{t = 1}$ are the equations we want to solve.

Consider the ideal $J$ generated by $\bB$ in $\KK(t)[x_1,\dots,x_n]$.
The roots of $J$ have coordinates in an algebraic closure of $\KK(t)$,
so we can view them in $\KKbar\langle \langle t \rangle \rangle^n$,
where $\KKbar\langle \langle t \rangle \rangle$ is the field of
Puiseux series with coefficients in $\KKbar$. Thus, these solutions
are meant to describe the local behaviour of the solutions of $\bB$ at
$t=0$. A vector $\balpha$ in $\KKbar\langle \langle t \rangle
\rangle^n$ admits a {\em valuation} $\nu(\balpha)$, defined as the
minimum of the valuations (with respect to $t$) of its coordinates, and we
say that $\balpha$ is {\em bounded} when $\nu(\balpha) \geq 0$.
This will be one of the conditions we impose on the solutions of $J$.

The algorithm is in essence a form of Newton iteration with respect to
$t$. One input needed for the algorithm is an upper bound $\varrho$ on
the precision in $t$ at which we need to do the computations. A
sufficient upper bound for $\varrho$ is the degree of the {\em
  homotopy curve}, which is the union of all dimension-1 irreducible
components of $V(\bB) \subset \KKbar{}^{n+1}$ whose projections on the
$t$-space are Zariski dense. In effect, this is the number of isolated
solutions of the system in $\KK[t,x_1,\dots,x_n]$ obtained by taking
all equations in $\bB$, together with a linear form in
$t,x_1,\dots,x_n$ with random coefficients.

Finally, as in~\cite{HSSV18}, the following proposition assumes that we are
given a {\em straight-line program} $\Gamma$ that computes the
polynomials $\bB$, that is, is a sequence of operations $+,-,\times$
that takes as input $t,x_1,\dots,x_n$ and evaluates $\bB$. Its {\em
  length} is simply the number of operations it performs.

\begin{proposition}\label{pro:start_mat} 
  Suppose that the following conditions hold:
  \begin{itemize}
  \item[$(i)$] the ideal generated by $\bB_{t =0}$ is radical and
    of dimension zero in $\KK[x_1,\dots,x_n]$, with $\chi$ solutions;
  \item[$(ii)$] all points in $V(\bB) \subset \KKbar\langle \langle t
    \rangle \rangle^n$ are bounded.
  \end{itemize}
  Then, the ideal $J$ generated by $\bB$ in $\KK(t)[x_1,\dots,x_n]$ is
  radical and of dimension zero, with $\chi$ solutions, and the system
  $\bB_{t =1}$ admits at most $\chi$ isolated solutions (counted with
  multiplicities).
  
  Furthermore, given a zero-dimensional parametrization of the solutions
  of $\bB_{t=0}$, a straight-line program $\Gamma$ of length $\beta$
  that computes $\bB$, and the upper bound $\varrho$ as above, there
  exists a randomized algorithm ${\sf Homotopy}$ which computes a
  zero-dimensional parametrization of the isolated solutions of $\bB_{t
    =1}$ using
  \begin{equation}\label{eq:aaa}
    \softO(\chi(\varrho + \chi^5) n^4 \beta)
  \end{equation}
  operations in $\KK$. 
\end{proposition}

\section{Main algorithm}\label{sec:colum_homotopy}

Given $\g=(g_1,\dots,g_s)$ and $\F=[f_{i,j}]_{1\le i \le p,\ 1 \le j
  \le q}$ as in Section~\ref{sec:homotopy}, our goal in this section
is to specify the homotopy algorithm. We design a suitable start
system for the symbolic homotopy algorithm, and we establish that this
system satisfies the assumptions of Proposition~\ref{pro:start_mat}.
The cost analysis is done in the next section.

In order to build the polynomials $\r=(r_1,\dots,r_s)$
of~\eqref{eqdef:u}, we take polynomials with the same supports at
$\g=(g_1,\dots,g_s)$ and generic coefficients, taking care to add the
constant $1$ to their monomial supports if it is missing. The main new
ingredient is the determination of the start matrix $\M$
of~\eqref{eqdef:V}. In this paper, we focus on what we call the {\em
  column support homotopy} where the construction of $\M$ is derived
from the unions of the supports of the entries of $\F$ per
columns. This extends a similar construction given in~\cite{HSSV18}
for dense polynomials, but which was instead based on the total
degrees of the columns of $\F$.


\subsection{Column support homotopy}
\label{subsec:sparse_col}
      
For $1\leq i\leq s$, let $\calA_i \subset \N^n$ denote the support of
$g_i$, to which we add the origin ${\bf 0} \in \N^n$. For $1\leq j
\leq q$, let $\calB_j  \subset \N^n$ be the {\em union} of the supports of the
polynomials in the $j$-th column of $\F$, to which we add $\bf 0$ as
well.

For  given $i$ and $j$ we  denote by $\kappa_i$ the
cardinality of $\calA_i$ and by $\mu_j$ the cardinality of $\calB_j$,
and let $ (n_{i,
  1}, \dots, n_{i, \kappa_i})$ and $( m_{j, 1}, \dots, m_{j, \mu_j})$
denote the monomials in $x_1,\dots,x_n$ supported by $\calA_i$ and
$\calB_j$, respectively. We can then  define the ``generic'' polynomials
supported on $\calA_1,\dots,\calA_s$ and $\calB_1,\dots,\calB_q$:
\begin{equation*} \label{eq:generic_polys}
  \frkr_i = \sum_{k=1}^{\kappa_i} \, \frkd_{i, k} n_{i, k}  \ \ (1 \le i \le s) \quad
  \quad {\rm and} \quad \frkm_{j} = \sum_{k=1}^{\mu_j} 
  \, \frke_{j, k} m_{j, k} \ \ (1 \le j \le q),
\end{equation*}
where all $\frkd_{i, k}$ and $\frke_{j, k}$ are new
indeterminates. Let $\frkc_{i, j}$, for $1 \leq i \leq p$ and $1 \leq
j \leq q$, be $p q$ additional new indeterminates so that $\frkA = \{
(\frkd_{i, k})_{1 \le i \le s, 1\le k \le \kappa_i} , (\frke_{j,
  k})_{1 \le j \le q, 1 \le k \le \mu_j} , (\frkc_{i, j})_{1 \le i \le
  p, 1 \le j \le q} \}$, the set of all these new indeterminates, has
size
\[ 
N = \sum_{i=1}^s \, \kappa_i + \sum_{i=1}^q \mu_i + pq. 
\] 
We then define the matrix
\begin{equation*} \label{eq:generic_mat_column} 
  \frkM = \begin{pmatrix}
    \frkc_{1, 1}\, \frkm_1 & \frkc_{1, 2} \, \frkm_2   & \hdots &
    \frkc_{1, q}\, \frkm_q \\  
    \vdots & \vdots &      & \vdots \\   
    \frkc_{p, 1}\, \frkm_1 & \frkc_{p, 2}\, \frkm_2 & \hdots &
    \frkc_{p, q}\, \frkm_q 
  \end{pmatrix} \in \KK[\frkA][x_1,\dots,x_n]^{p \times q}.  
\end{equation*}
As before, for $\rho$ in $\KKbar{}^N$, for any polynomial $f$ having
coefficients in $\KKbar[\frkA]$, $\Theta_\rho(f)$ is the polynomial
with coefficients in $\KKbar$ obtained through evaluation of the
indeterminates $\frkA$ at $\rho$; the notation carries over to
polynomial matrices.

We will use $\frkM$ and $\frkr = (\frkr_1, \ldots, \frkr_s)$ to
construct our start system, by assigning random values to all
indeterminates in $\frkA$. Thus, we let $t$ be a new indeterminate and
we denote by $\frkB$ the polynomials in $\KK[\frkA][t,x_1,\dots,x_n]$
obtained by considering the equations $(1-t) \cdot \frkr + t \cdot \g$
and the $p$-minors of $(1-t)\cdot \frkM+t\cdot \F$. Our goal in this
section is to establish the following result.
\begin{proposition}\label{prop:hyp}
  There exists a non-empty Zariski open subset $\Open$ of
  $\KKbar{}^N$ such that for $\rho$ in $\Open$,
  $\bB:=\Theta_\rho(\frkB)$ satisfies the assumptions of
  Proposition~\ref{pro:start_mat}.
\end{proposition}
In other words, we will prove that, for such a choice of $\rho$, the ideal
generated by $\bB_{t=0}$ in $\KKbar[x_1, \ldots, x_n]$ is radical and
zero-dimensional (this is done in the next subsection) and that the solutions of
$\bB$ in $\KKbar\langle\langle t \rangle\rangle^n$ are bounded. This boundedness
properties is proved in Subsection~\ref{ssec:boundedness} using properties of
Lagrange type systems which are established in Subsection~\ref{ssec:lagrange}.

Note also the following consequence of Proposition~\ref{pro:start_mat}: the
number of isolated solutions of the system we want to solve (counting
multiplicities) is bounded above by the number of solutions of a generic start
system $\Theta_\rho(\frkB)_{t=0}$.


\subsection{Properties of the start system}
\label{subsec:multi}
\label{subsec:column_cost}

In this subsection, we prove that for a generic choice of $\rho$ in
$\KKbar{}^N$, if we write $\bB:=\Theta_\rho(\frkB)$ then the ideal
generated by $\bB_{t = 0}$ in $\KKbar[x_1, \ldots, x_n]$ is radical
and zero-dimensional.

\begin{proposition} \label{prop:start_pro}
  There exists a non-empty Zariski open set $\Openstart \subset
  \KKbar{}^N$ such that for $\rho$ in $\Openstart$, writing
  $\bB:=\Theta_\rho(\frkB)$, the ideal generated by $\bB_{t = 0}$ in
  $\KKbar[x_1, \ldots, x_n]$ is radical of dimension zero.
\end{proposition}
\begin{proof}
Note first that the equations  $\bB_{t = 0}$ that we are considering are the $p$-minors of
$\Theta_\rho(\frkM)$, together with
$\Theta_\rho(\frkr_1,\dots,\frkr_s)$. Now, any $p$-minor
of $\frkM$ has the form $\frak C_{i_1,\dots,i_p} \frkm_{i_1} \cdots
\frkm_{i_p}$, for some choice of columns $i_1,\dots,i_p$, where $\frak
C_{i_1,\dots,i_p}$ is the determinant
\begin{equation*} 
  \frak C_{i_1,\dots,i_p} = \left | \begin{matrix}
    \frkc_{1, i_1} & \frkc_{1, i_2} & \hdots &  \frkc_{1, i_p} \\  
    \vdots & \vdots &      & \vdots \\   
    \frkc_{p, i_1} & \frkc_{p, i_2} & \hdots &  \frkc_{p, i_p} 
  \end{matrix}\right | \in \KK[\frkA].
\end{equation*}
Our first constraint on $\rho$ is thus that $\Theta_\rho(\frak
C_{i_1,\dots,i_p}) \in \KKbar$ is non-zero for all
$\{i_1,\dots,i_p\}$. In this case, a point $\balpha$ in $\KKbar{}^n$
cancels all the $p$-minors of $\Theta_\rho(\frkM)$ if and only if it
cancels all products $\Theta_\rho(\frkm_{i_1}) \cdots
\Theta_\rho(\frkm_{i_p})$. This is the case if and only if there
exists $\j = \{i_1, \dots, i_{q-p+1}\} \subset \{1, \dots, q\}$ such
that $\Theta_\rho(\frkm_{i_1}),\dots,\Theta_\rho(\frkm_{i_{q-p+1}})$
all vanish at $\balpha$.  

Since we assume $n=q-p+s+1$, we can rewrite $q-p+1$ as $n-s$.  Then,
for a subset $\j=\{i_1, \dots, i_{n-s}\} \subset \{1, \dots, q\}$,
consider the polynomials $\frkM_{\j} = (\frkm_{i_1}, \dots,
\frkm_{i_{n-s}})$. By Proposition~\ref{prop:generic_single}(i), there
exists a non-empty Zariski open set $\mathscr{O}_{\j} \subset
\KKbar{}^N$ such that for $\rho$ in $\mathscr{O}_{\j}$, the ideal
generated by $\Theta_\rho(\frkM_\i, \frkr)$ is radical and admits
finitely many solutions. For subsets $\j'$ and $\j$ of $\{1, \dots,
q\}$ of cardinalities $n-s$ such that $\j \ne \j'$, the system defined
by $\frkM_{\j \cup \j'}$ and $\frkr$ contains at least $n+1$
polynomials in $\KK[\frkA][x_1, \dots, x_n]$. By using
Proposition~\ref{prop:generic_single}(ii), there exists a non-empty
Zariski open set $\mathscr{O}_{\j\cup \j'} \subset \KKbar{}^N$ such
that for $\rho$ in $\mathscr{O}_{\j\cup \j'}$, the system
$\Theta_{\rho}(\frkM_{\j \cup \j'},\frkr)$ has no solutions in
$\KKbar{}^n$.

Taking the intersection of these $\mathscr{O}_{\j}$ and
$\mathscr{O}_{\j\cup \j'}$ (which are finite in number), together with
the condition that the determinants $\Theta_\rho(\frak
C_{i_1,\dots,i_p})$ do not vanish, defines a non-empty Zariski open
$\Openstart \subset \KKbar{}^N$.  Thus, for $\rho$ in $\Openstart$, the sets
$V(\Theta_\rho(\frkM_{\j}, \frkr))$, for any subset $\j$ of $\{1,
\dots, q\}$ of cardinality $n-s$, are finite and pairwise disjoint,
and their union is $V(\bB_{t=0})$. In particular, the latter set is
finite.

Take $\rho$ in $\Openstart$ and $\balpha$ in $V(\bB_{t=0})$. We now prove
that the ideal generated by $\bB_{t = 0}$, that is, by the $p$-minors
of $\Theta_\rho(\frkM)$ and $\Theta_\rho(\frkr_1,\dots,\frkr_s)$, has
multiplicity one at $\balpha$. This will imply that $\bB_{t = 0}$
generates a radical ideal. For this, we will use the fact that
$\balpha$ is the root of the system $\Theta_\rho(\frkM_{\j}, \frkr)$,
for a unique subset $\j=(i_1,\dots,i_{n-s})$ of $\{1, \dots, q\}$ of
cardinality $n-s$, and that $ \Theta_\rho(\frkM_{\j}, \frkr) $ has
multiplicity one at $\balpha$.

Let then $\bm j=(j_1,\dots,j_{p-1})$ denote the $q - (n-s) = p-1$
columns of $\frkM$ not indexed by $\j$.  For $i$ in $\j$, the equation
$\Theta_\rho(\frak C_{j_1,\dots,j_{p-1},i} \frkm_{j_1}\cdots
\frkm_{j_{p-1}} \frkm_i)$ appears among the generators of $\bB_{t =
  0}$. In the local ring at $\balpha$, we can divide by the non-zero
quantity $\Theta_\rho(\frak C_{j_1,\dots,j_{p-1},i} \frkm_{j_1}\cdots
\frkm_{j_{p-1}})(\balpha)$.  This implies that locally at $\balpha$,
$\bB_{t = 0}$ is generated by the polynomials
$\Theta_\rho(\frkm_{i_1}),\dots,\Theta_\rho(\frkm_{i_{n-s}})$ and
$\Theta_\rho(\frkr)$. The conclusion follows.
\end{proof}


\subsection{The associated Lagrange system}\label{ssec:lagrange}

To establish the boundedness property, since $\frkB$ is
overdetermined, it will be convenient to introduce new variables
$\bell = (\ell_1, \ldots, \ell_p)$ and to work with the {\em Lagrange
system}, which consits of $s+q+1$ equations defined by
\begin{equation} \label{eq:slack_eqs_column} (1-t) \frkr + t \g\ = \ [\ell_1
    ~\cdots~ \ell_p] ((1-t) \frkM + t \F)\ = \ \frkt_1\ell_1 + \cdots + \frkt_p
  \ell_p -1 = 0,
\end{equation} 
where $\frkt = (\frkt_1, \ldots,\frkt_{p})$ are new indeterminate
coefficients. Recall that $n = q-p+s+1$, so $s+q+1 =
  n+p$; we will write these equations as $\frak
  H=(\frak H_1, \ldots, \frak H_{n+p})$.

There are now $N+p$ parameters in these equations, with elements of the
parameter space $\KKbar{}^{N+p}$  written as $\sigma=(\rho,\tau)$,
with $\rho$ in $\KKbar{}^N$ and $\tau$ in $\KKbar{}^p$. For $\sigma$
in $\KKbar{}^{N+p}$ and $f$ a polynomial with coefficients in
$\KK[\mathfrak{A}, \frkt]$, we write as usual $\Theta_{\sigma}(f)$ for
the polynomial whose coefficients are obtained from those of $f$, with
$\frak A$ evaluated at $\rho$ and $\frkt$ evaluated at $\tau$. As
before, the notation carries over to vectors or matrices of
polynomials. 

For $1\leq i \leq n+p$, $\frak H_i$ can be decomposed as $\frak H_i =
{\eta_{i}} + t {\frak h_{i}}$ with both ${\eta_{i}}$ and $\frak h_i$ in
$\KK[\mathfrak{A}, \frkt][\x, \bell]$. In particular, note that the
polynomials $\bm \eta=(\eta_1,\dots,\eta_{n+p})$ form the Lagrange
system
\begin{equation*}  \frkr_1 = \cdots = \frkr_s \ = \ [\ell_1
    ~\cdots~ \ell_p] \frkM \ = \ \frkt_1\ell_1 + \cdots + \frkt_p
  \ell_p +1 = 0
\end{equation*} 
in $\KK[\mathfrak{A}, \frkt][\x,\bell]$, so for $i=1,\dots,q$, the
polynomial $\eta_{s+i}$ is $(\frkc_{1,i} \ell_1 + \cdots + \frkc_{p,i}
\ell_p)\frkm_i$.

In what follows, we discuss properties of the polynomials
$\Theta_{\sigma}(\bm\eta)$ and their initial forms
$\init_\e(\Theta_{\sigma}(\bm\eta))$, for $\e$ in $\Q^{n+p}$.  Our
first claim is the following; the proof is straightforward.
\begin{lemma}
  For $\sigma$ in $(\KKbar{}-\{0\})^{N+p}$ and $\e$ in $\Q^{n+p}$,
  $\init_\e(\Theta_{\sigma}(\bm\eta))=\Theta_{\sigma}(\init_\e(\bm\eta))$.
\end{lemma}
The second proposition uses the specific shape of the equations $\frak
H$ to derive information about their roots.
\begin{proposition}\label{prop:inite}
  Let $\bm \phi = (t^{e_1}c_1+\dots, \ldots,
  t^{e_{n+p}}c_{n+p}+\dots)$ be in $\KKbar\langle\langle t
  \rangle\rangle^{n+p}$ with, for all $i=1,\dots,n+p$, $e_i$ in $\Q$
  and $c_i$ in $\KKbar-\{0\}$.

  Then for $\sigma$ in $(\KKbar{}-\{0\})^{N+p}$, we have the
  following: if $\bm \phi$ cancels $\Theta_{\sigma}(\frak H)$, then $\bm
  c=(c_1,\dots,c_{n+p})$ cancels $\Theta_{\sigma}(\init_\e(\bm\eta))$,
  with $\e = (e_1,\dots,e_{n+p})$.
\end{proposition}
\begin{proof} 
  For $i=1,\dots,s$, we have $\frak H_i= \frkr_i + t(g_i-\frkr_i)$, so
  $\eta_i=\frkr_i$ and $\frak h_i=g_i-\frkr_i$. Thus by construction,
  the monomial support of $\frak h_i$ (with respect to
  $x_1,\dots,x_n,\ell_1,\dots,\ell_p$) is the same as that of
  $\frkr_i$. This means that for any term $k x_1^{u_1} \cdots
  \ell_p^{u_{n+p}}$ in $\frak h_i$, with $k$ in $\KK[\frak A]$, there
  exists a term $k' x_1^{u_1} \cdots \ell_p^{u_{n+p}}$ in $\eta_i$,
  where $k'$ is one of the indeterminates $\frkd_{i,j}$.

  Take $\sigma$ as in the statement of the proposition, and write
  $a=\Theta_{\sigma}(\frak H_i)$, $b=\Theta_{\sigma}(\eta_i)$ and
  $c=\Theta_{\sigma}(\frak h_i)$, so that $b(\bm \phi) + t c(\bm \phi) =a(\bm \phi)=0$.
  Using our assumption on $\sigma$, we deduce that for any term of the form
  $k t \phi_1^{u_1} \cdots \phi_{n+p}^{u_{n+p}}$ appearing in $t c(\bm
  \phi)$, there is a term $k' \phi_1^{u_1} \cdots \phi_{n+p}^{u_{n+p}}$
  appearing in $b(\bm \phi)$, with non-zero coefficient $k'$. In
  particular, all terms of smallest valuation in $a(\bm \phi)$ appear
  in $b(\bm \phi)$, and must add up to zero.  Taking their first
  coefficient, this implies that $\bm c$ cancels $\init_\e(b)$.
  
  The proof for the polynomials $\frak H_{s+1},\dots,\frak H_{s+q}$,
  $\eta_{s+1},\dots,\eta_{s+q}$ and $\frak h_{s+1},\dots,\frak
  h_{s+q}$ is similar, taking into account that
  $\eta_{s+i}=(\frkc_{1,i} \ell_1 + \cdots + \frkc_{p,i}
  \ell_p)\frkm_i$. Indeed, again, for $i=1\dots,q$, the monomial
  support of $\frak h_{s+i}$ is the same as that of $\eta_{s+i}$; if
  we define $a,b,c$ as above, our assumption that no entry of $\sigma$
  vanishes implies as before that all terms of smallest valuation in
  $a(\bm \phi)$ appear in $b(\bm \phi)$, and add up to zero. Finally,
  for $\frak H_{s+q+1}=\frak H_{n+p}$, we have that $\frak h_{n+p}=0$,
  and the claim follows as above.
\end{proof}

Our last property requires a longer proof. For generic choices of
$\sigma$, it constrains the possible roots of the system
$\Theta_{\sigma}(\init_\e(\bm \eta))$ introduced in the previous
proposition.

\begin{proposition}\label{prop:gen3}
  There exists a non-empty Zariski open set $\Openinit\subset
  \KKbar{}^{N+p}$ such that for $\sigma \in \Openinit$, the
  following holds for any $\e$ in $\Q^{n+p}$: for $j=1,\dots,n+p$, the
  system obtained by setting the $j$-th variable to $1$ in
  $\Theta_{\sigma}(\init_\e(\bm \eta))$ has no solution in
  $(\KKbar-\{0\})^{n+p-1}$.
\end{proposition}
\begin{proof}
  Even though there is an infinite number of vectors $\e$ to take into
  account, there is only a finite number of possible systems
  $\init_\e(\bm \eta)$. Thus, in what follows, we assume $\e$ is fixed
  and prove the existence of a suitable Zariski open set, knowing that
  we will eventually take the intersection of the open sets
  corresponding to the finite number of systems $\init_\e(\bm \eta)$.
  Similarly, without loss of generality, we assume $j=1$, so
  that we are setting $x_1$ to $1$.

  Thus, we call $ \bar{\bm \eta}=(\bar \eta_1,\dots,\bar \eta_{n+p})$
  the polynomials in $\KK[\frak A,\frak
    t][x_2,\dots,x_n,\ell_1,\dots,\ell_p]$ obtained by setting $x_1$
  to $1$ in $\init_\e(\bm\eta)$. We will prove that for a generic
  $\sigma$ in $\KKbar{}^{N+p}$, the system $\Theta_{\sigma}(\bar{ \bm
    \eta}) \subset \KKbar[x_2,\dots,x_n,\ell_1,\dots,\ell_p]$ has no
  solution in $(\KKbar-\{0\})^{n+p-1}$ (this system is indeed the one
  mentioned in the statement of the proposition, since $\Theta_\sigma$
  and variable evaluation commute).

  For $i=1,\dots,n+p$, denote by $\frak S_i$ the subset of $(\frak
  A,\frak t)$ consisting of those indeterminates that appear in the
  coefficients of $\eta_i$ (so it also contains those that appear in
  the coefficients of $\bar \eta_i$). With this convention, the sets
  $\frak S_i$ are pairwise disjoint, and $(\frak S_1,\dots,\frak
  S_{n+p})$ is the set of all indeterminate coefficients $(\frak
  A,\frak t)$ that appear in $\bm \eta$. For all $i$, we let $t_i$ be
  the cardinality of $\frak S_i$, and we will write the elements of
  $\KKbar{}^{t_i}$ as $\rho_i$, so that  a vector $\sigma \in
  \KKbar{}^{N+p}$ can be decomposed as $\sigma=(\rho_1,\dots,\rho_{n+p})$.
  Given $(\rho_1,\dots,\rho_i)$ in $\KKbar{}^{t_1+\cdots+t_i}$,
  $\Theta_{(\rho_1,\dots,\rho_i)}$ denotes as usual the mapping that
  evaluates the $t_1+\cdots+t_i$ indeterminates $\frak S_1,\dots,\frak
  S_i$ at $(\rho_1,\dots,\rho_i)$.

  The key property we will use below is the following: {\em for any
    $\balpha$ in $(\KKbar-\{0\})^{n+p-1}$, the polynomial $\gamma \in
    \KKbar[\frak S_i]$ obtained by evaluating
    $x_2,\dots,x_n,\ell_1,\dots,\ell_p$ at the coordinates of $\bm
    \alpha$ in $\bar \eta_i$ is non-zero.}  For $i=1,\dots,s$ and
  $i=n+p$, this is because the coefficients of $\bar \eta_i$ are sums
  of elements of $\frak S_i$, no element in $\frak S_i$ appears in two
  such coefficients, and all coordinates of $\bm \alpha$ are non-zero.
  For $i=s+1,\dots,n+p-1$, since $\eta_{i}$ is $(\frkc_{1,i-s} \ell_1
  + \cdots + \frkc_{p,i-s} \ell_p)\frkm_{i-s}$, its initial form
  $\init_\e(\eta_i)$ is the product $\init_\e(\frkc_{1,i-s} \ell_1 +
  \cdots + \frkc_{p,i-s} \ell_p)\init_\e(\frkm_{i-s})$.  After setting
  $x_1$ to $1$, we deduce that $\bar \eta_i$ factors as $\bar \eta_i=
  f_i g_i$, where the coefficients of both $f_i$ and $g_i$ are sums of
  elements of $\frak S_i$, and again, no element in $\frak S_i$
  appears in two such coefficients. Thus, the evaluations of $f_i$ and
  $g_i$ at $\bm \alpha$ are non-zero, and the same holds for $\bar
  \eta_i$.

  To describe algebraic sets in the torus $(\KKbar-\{0\})^{n+p-1}$, we
  work in $\KKbar{}^{n+p}$, using a new indeterminate $Z$ and taking
  into account the relation $x_2 \cdots x_n \ell_1 \cdots \ell_p
  Z=1$.  Then, for $i=0,\dots,n+p$, we will
  prove the following: {\em for a generic choice of $(\rho_1,\dots,\rho_i)$
    in $\KKbar{}^{t_1+\cdots+t_i}$ (in the Zariski sense), the
    zero-set of $\Theta_{(\rho_1,\dots,\rho_i)}(\bar
    \eta_1,\dots,\bar \eta_i)$ and $x_2 \cdots x_n
    \ell_1 \cdots \ell_p Z-1$ has dimension at most $n+p-1-i$ in
    $\KKbar{}^{n+p}$.} Taking $i=n+p$ proves our claim.

  The proof is by induction on $i$. For $i=0$, there is nothing to
  prove, so let us assume that our claim holds for $i-1$ (for some
  index $i \ge 1$), and prove that it holds at index $i$.  We proceed
  by contradiction, assuming our claim does not hold.  In this case,
  the vectors $(\rho_1,\dots,\rho_i)$ for which the zeros of
  $\Theta_{(\rho_1,\dots,\rho_i)}(\bar \eta_1,\dots,\bar \eta_i)$ and $x_2
  \cdots x_n \ell_1 \cdots \ell_p Z-1$ have dimension at most
  $n+p-1-i$ in $\KKbar{}^{n+p}$ are contained in a hypersurface of the
  parameter space $\KKbar{}^{t_1 + \cdots + t_i}$. Thus  they satisfy a
  relation $P(\rho_1,\dots,\rho_i) = 0$, for some non-zero polynomial $P$ in
  $\KKbar[\frak S_1,\dots,\frak S_i]$. Then, take $(\rho_1,\dots,\rho_{i-1})$ in
  $\KKbar{}^{t_1 + \cdots + t_{i-1}}$ such that
  \begin{itemize}
  \item $P(\rho_1,\dots,\rho_{i-1},\frak S_i) \in \KKbar[\frak S_i]$ is not
    identically zero;
  \item the zero-set $V$ of $\Theta_{(\rho_1,\dots,\rho_{i-1})}(\bar
    \eta_1,\dots,\bar \eta_{i-1})$ and $x_2 \cdots x_n \ell_1 \cdots
    \ell_p Z-1$ has dimension at most $n+p-i$ in $\KKbar{}^{n+p}$
    (this is possible by the induction assumption). By Krull's
    theorem, all its irreducible components have dimension exactly
    $n+p-i$.
  \end{itemize}
  The first condition implies that for a generic $\rho_i$ in
  $\KKbar{}^{t_i}$, the zero-set of $\Theta_{(\rho_1,\dots,\rho_i)}(\bar
  \eta_1,\dots,\bar \eta_i)$ and $x_2 \cdots x_n
  \ell_1 \cdots \ell_p Z-1$ has dimension at least
  $n+p-i$. Equivalently, this means that intersection of $V$ and
  $\Theta_{(\rho_1,\dots,\rho_i)}(\bar \eta_i)$ has dimension $n+p-i$. Let us see
  how to derive a contradiction.
  
  Let $V_1,\dots,V_d$ be the irreducible components of $V$. Pick $\bm
  \alpha_1$ in $V_1$, \dots, $\bm \alpha_d$ in $V_d$, and let
  $\gamma_1,\dots,\gamma_d$ be the polynomials in $\KKbar[\frak S_i]$
  obtained by evaluating $x_2,\dots,x_n,\ell_1,\dots,\ell_p$ at the
  coordinates of  $\bm \alpha_1,\dots,\bm \alpha_d$, respectively, in
  $\bar \eta_i$. As we pointed out above, all $\gamma_i$'s are
  non-zero, and thus so is $\Gamma:= \gamma_1\cdots\gamma_d \in
  \KKbar[\frak S_i]$. In particular, for a generic choice of $\rho_i$ in
  $\KKbar{}^{t_i}$, $\Theta_{(\rho_1,\dots,\rho_i)}(\bar \eta_i)$ vanishes
  at none of $\bm \alpha_1,\dots,\bm \alpha_d$, and so it intersects each
  $V_i$ (and thus $V$) in dimension $n+p-i-1$. This contradicts the
  previous paragraph.
\end{proof}


\subsection{Boundedness property}\label{ssec:boundedness}

Using the results in the previous subsection, we finally establish the
second property needed for our homotopy algorithm: we prove that for a
generic $\rho$ in $\KKbar{}^N$, the solutions of
$\bB=\Theta_\rho(\frkB)$ in $\KKbar\langle\langle t \rangle\rangle^n$
are bounded.
    
\begin{proposition}\label{prop:boundedness}
  There exists a non-empty Zariski open set $\Openbounded\subset
  \KKbar{}^{N}$ such that for $\rho\in \Openbounded$, writing
  $\bB:=\Theta_\rho(\frkB)$, all points in $V(\bB)\subset
  \KKbar\langle\langle t \rangle\rangle^n$ are bounded.
\end{proposition} 
\begin{proof}
  By Proposition~\ref{prop:gen3}, there exists a non-empty Zariski
  open set $\Openinit\subset \KKbar{}^{N+p}$ such that for any
  $\sigma=(\rho,\tau)$ in $\Openinit$, the following holds: for any
  $\e$ in $\Q^{n+p}$ and any $j$ in $\{1,\dots,n+p\}$, the system
  obtained by setting the $j$-th variable to $1$ in
  $\Theta_{\sigma}(\init_\e(\bm \eta))$ has no solution in
  $(\KKbar-\{0\})^{n+p-1}$.

  We then let $\Openinit' \subset \KKbar{}^N$ be the image of
  $\Openinit$ through the projection $\pi: \sigma =(\rho,\tau) \mapsto
  \rho$; this is a non-empty Zariski open. Finally, we let
  $\Openbounded$ be the intersection of $\Openinit'$ with
  $(\KKbar{}-\{0\})^{N} \subset \KKbar{}^N$. We take $\rho$ in
  $\Openbounded$ and we prove that all solutions of
  $\Theta_{\rho}(\frkB)$ in $\KKbar\langle\langle t \rangle\rangle^n$
  are bounded.

  Take such a solution, and write it
  $\balpha=(\alpha_1,\dots,\alpha_n) \in \KKbar\langle\langle t
  \rangle\rangle^n$. By construction, there exists a non-zero
  $(\lambda_1,\dots,\lambda_p)\in \KKbar\langle\langle t
  \rangle\rangle^p$ such that $[\lambda_1~\cdots~\lambda_p]$ is in the
  left nullspace of $\frkM(\balpha)$. Let $v \in \Q$ be the valuation
  of this vector, and let $(\lambda'_1,\dots,\lambda'_p) \in
  \KKbar{}^p$ be the vector of coefficients of $t^v$ in
  $(\lambda_1,\dots,\lambda_p)$, so that
  $(\lambda'_1,\dots,\lambda'_p)$ is not identically zero.  Let us
  then take $\tau=(\tau_1,\dots,\tau_p)$ such that
  $\sigma:=(\rho,\tau)$ is in $\Openinit$ and in addition
  $\tau_1 \ne 0,\dots,\tau_p\ne 0$ and $ \tau_1 \lambda'_1 + \cdots +
  \tau_p \lambda'_p \ne 0$ (this is possible, since all these conditions 
  are Zariski-open). In particular, $ \tau_1 \lambda_1 + \cdots
  + \tau_p \lambda_p \ne 0$.  We can then define
  $\bar {\bm \lambda}=(\bar\lambda_1,\dots,\bar\lambda_p)$ by $\bar\lambda_i = \lambda_i
  /( \tau_1 \lambda_1 + \cdots + \tau_p \lambda_p) $ for all $i$.  Let
  us write $\bm \phi = (\balpha,\bar {\bm\lambda})$; our goal is then to
  prove that $\bm \phi$ is bounded, since it will imply that $\balpha$
  is bounded.

By construction, the vector $[\bar\lambda_1~\cdots~\bar\lambda_p]$ is
still in the left nullspace of $\frkM(\balpha)$ and satisfies
$\tau_1\bar\lambda_1 + \cdots + \tau_p \bar\lambda_p -1 = 0$. Hence, the vector
$\bm \phi$ is in $V(\Theta_\sigma(\frak H))$. Let us then write $\bm
\phi = (t^{e_1}c_1+\dots, \ldots, t^{e_{n+p}}c_{n+p}+\dots)$ with, for
all $i=1,\dots,n+p$, $e_i$ in $\Q$ and $c_i$ in $\KKbar-\{0\}$. Because
none of the coordinates of $\sigma$ vanishes, we can apply
Proposition~\ref{prop:inite}, and deduce that $\bm
c=(c_1,\dots,c_{n+p})$ cancels $\Theta_{\sigma}(\init_\e(\bm\eta))$,
with $\e = (e_1,\dots,e_{n+p})$.

Suppose then by way contradiction that some $e_i$ is negative; without
loss of generality, we can assume that $e_1 < 0$. The polynomials
$\Theta_{\sigma}(\init_\e(\bm\eta))$ are weighted-homogeneous, for the
weight vector $\e$. In particular, the point
$$\tilde {\bm c} = \left (1, \frac{c_2}{\epsilon^{e_2}}, \ldots,
\frac{c_{n+p}}{\epsilon^{e_{n+p}}} \right )$$ is also a solution of
these equations, where $\epsilon$ denotes any element in $\KKbar$ such
that $\epsilon^{e_1} = c_1$.  Note that none of the coordinates of the
vector $\tilde{\bm c}$ vanishes.  However, by construction, $\sigma$
is in $\Openinit$, so Proposition~\ref{prop:gen3} asserts that the
system obtained by setting the first variable $x_1$ to $1$ in
$\Theta_{\sigma}(\init_\e(\bm \eta))$ has no solution in
$(\KKbar-\{0\})^{n+p-1}$. This is the contradiction we wanted, so we
have $e_i \ge 0$ for all $i$, as claimed.
\end{proof}

At this stage, to prove Proposition~\ref{prop:hyp}, it suffices to let
$\Open$ be the intersection of $\Openstart$ (from
Proposition~\ref{prop:start_pro}) and $\Openbounded$ (from the
proposition above). 

\section{Cost analysis}\label{sec:complexity}

Let the polynomials in $\g=(g_1,\dots,g_s)$ and $\F=[f_{i,j}]_{1\le i
  \le p,\ 1 \le j \le q}$ be as before. To find the isolated points in 
$ V_p(\F, \g) $, we take $\bm B = \Theta_\rho(\frkB)$ as in the previous 
section, for a randomly chosen $\rho \in \KK^N$ and apply the 
{\sf Homotopy} algorithm of Proposition~\ref{pro:start_mat}.

Proposition~\ref{prop:hyp} established the basic properties needed for
the correctness of our homotopy algorithm. To finish the analysis, and
establish a cost bound, we now give upper bounds on the parameters
that appear in the runtime reported in
Proposition~\ref{pro:start_mat}, such as the size of the input, the
number of solutions to our start system and on the degree of the
homotopy curve; we also have to give the cost of solving the start
system.

We first consider the case of arbitrary sparse polynomials, for which
we state our results in terms of certain mixed volumes; later we
discuss the particular case of weighted-degree polynomials. Some
quantities will be defined similarly in both cases.  As before, for
$i=1,\dots,s$, $\calA_i \subset \N^n$ denotes the support of $g_i$, to
which we add the origin ${\bf 0} \in \N^n$, and for $j=1,\dots,q$,
$\calB_j \subset \N^n$ is the union of the supports of the polynomials
in the $j$-th column of $\F$, to which we add $\bf 0$ as well. For
indices $i,j$ as above, we let $a_i$, respectively $b_j$, be the
cardinality of $\calA_i$, respectively\ $\calB_j$.  As input, in
either case, we are given $\g$ and $\F$ through the list of their
non-zero terms; this involves $O(\gamma)$ elements in $\KK$, with
\begin{equation}\label{eqdef:gamma}
\gamma:=a_1 + \cdots + a_s + p(b_1 + \cdots + b_q).  
\end{equation}
Finally, we let $d$ be the maximum degree of all the polynomials in $\bm g$
and $\bm F$.


\subsection{General sparse polynomials}
\label{ssec:degreebound}

\paragraph{Representing the input.} 
The algorithm in Proposition~\ref{pro:start_mat} takes as input a
straight-line program representation of the polynomials
$\bB=\Theta_\rho(\frkB)$. To obtain such a straight-line program is
straightforward. We first compute the values of all monomials
supported on $\calA_1,\dots,\calA_s,\calB_1,\dots,\calB_q$; we then
combine them to obtain the polynomials $(1-t) \cdot \Theta_\rho(\frkr)
+ t \cdot \g$ and the matrix $(1-t)\cdot \Theta_\rho(\frkM)+t\cdot
\F$, and take all $p$-minors in this matrix. 

Computing the value of a single monomial supported on $\calA_i$,
respectively\ $\calB_j$, can be done through repeated squaring, using $O(n
\log(d))$ operations in $\KK$. Hence, we can obtain the values of all
monomials supported on $\calA_1,\dots,\calA_s,\calB_1,\dots,\calB_q$ by
using a straight-line program of length $O(n \gamma \log(d))$. Combining
these monomials to obtain $(1-t) \cdot \Theta_\rho(\frkr) + t \cdot
\g$ and $(1-t)\cdot \Theta_\rho(\frkM)+t\cdot \F$ takes another
$O(\gamma)$ operations. Finally, it takes $O(p^4 {q \choose p} )$
operations to compute all $p$-minors of the latter matrix using a
division-free determinant algorithm. Altogether, we obtain a
straight-line program of length
\begin{equation}\label{eqdef:betasparse}
\beta \in 
O\left(n \gamma \log(d) +    p^4 {q \choose p} \right )
\end{equation}
to compute all entries of $\bB$.

\paragraph{Number of solutions of the start system.}
For $\rho$ in the open set $\Open \subset \KKbar{}^N$ defined in
Proposition~\ref{prop:hyp}, we saw that the solutions of the
start system $\bB_{t=0}$ are the disjoint union of the solutions of
the systems $\Theta_\rho(\frkM_{\j},\r)$, where for a subset
$\j=\{i_1, \dots, i_{n-s}\}$ of $\{1, \dots, q\}$ we write $\frkM_{\j}
= (\frkm_{i_1}, \dots, \frkm_{i_{n-s}})$.

For $i=1,\dots, s$ and $j=1,\dots, q$, we let $\calC_i$ and
$\calD_j$ be the convex hulls of respectively $\calA_i$ and $\calB_j$.
Proposition~\ref{theorem:bers75} then implies that, for $\j$ as above,
the number of solutions of $\Theta_\rho(\frkM_\j, \frkr)$ in
$\KKbar{}^n$ is  the mixed volume $$\chi_\j := {\sf
  MV}(\calC_1, \dots, \calC_s, \calD_{i_1}, \dots, \calD_{i_{n-s}})$$
for any $\rho$ in a certain non-empty Zariski open set ${\OpenBKK}_\j
\subset \KKbar{}^N$.  Define
\begin{equation} \label{eq:zzz}
\chi := \sum_{\i=\{i_1, \dots, i_{n-s}\} \subset \{1, \dots,
  q\}}\chi_\j= \sum_{\i=\{i_1, \dots, i_{n-s}\} \subset \{1, \dots,
  q\}} 
{\sf MV}(\calC_1, \dots, \calC_s, \calD_{i_1}, \dots,
\calD_{i_{n-s}}),
\end{equation}
and let $\Open'$ be the intersection of $\Open$ with the finitely many
${\OpenBKK}_\j$. Then, for $\rho$ in $\Open'$, the start system
$\bB_{t=0}$ has precisely $\chi$ solutions.  As we pointed out after
Proposition~\ref{prop:hyp}, this implies that the system $\bB_{t=1}$
which we want to solve admits at most $\chi$ isolated solutions,
counted with multiplicities.

\paragraph{Solving the start system.} 
To solve the systems $\Theta_\rho(\frkM_{\j},\frkr)$, we rely on the
sparse symbolic homotopy algorithm of~\cite[Section~5]{Jer09}. This
algorithm finds the solutions of a sparse system of $n$ equations in
$n$ unknowns, with arbitrary support and generic coefficients (in the
Zariski sense); this means that in addition to the constraint
$\rho\in\Open$, our choice of $\rho$ will also have to satisfy the
constraints stated in that reference.

The runtime of this algorithm depends on some combinatorial quantities
(we refer to the original reference for a more extensive discussion):
we need a so-called {\em lifting function} $\bm\omega_\i$, and the
associated {\em fine mixed subdivision} $M_\i$, for the support
$\calA_1,\dots,\calA_s,\calB_{i_1},\dots,\calB_{i_{n-s}}$ of $\frkr$ and
$\frkM_{\j}$~\cite{Huber95}. We then let $w_{\i}$ be the maximum value
taken by $\bm\omega_\i$ on the support, and $\mu_\i$ be the maximum norm
of the (primitive, integer) normal vectors to the cells of $M_\i$.
Then, the algorithm in~\cite[Theorem~6.2]{Jer09}  compute as
zero-dimensional parametrization $\scrR_{\j}$ such that $Z(\scrR_{\j})
= V(\Theta_\rho(\frkM_\j, \frkr))$ using $\softO(n^5 \gamma \log(d)
\chi_\i^2 \mu_\i w_\i )$ operations in $\KK$.

Taking the union of all these parametrizations, using for
example,~\cite[Lemma J.3]{SafeySchost17}, does not introduce any added
cost. Thus we obtain a randomized algorithm to compute a
zero-dimensional parametrization of $V_p(\Theta_\rho(\frkM,
\frkr))$ using
\begin{equation}\label{eq:wm}
\softO(n^5 \gamma \log(d) \chi^2 \mu w)  
\end{equation}
operations in $\KK$, where we write $\mu:=\max_\i(\mu_\i)$ and
$w:=\max_\i(w_\i)$.

\paragraph{Degree of the homotopy curve.}
The complexity of the ${\sf Homotopy}$ algorithm depends on $\chi$,
which measures the number of solutions which are tracked during the
homotopy, and on the precision $t^\varrho$ at which we need to do the
computations. As mentioned in Section~\ref{sec:homotopy}, an upper
bound for $\varrho$ is the number of isolated points defined by the
equations in $\bB=\Theta_\rho(\frak B)$ together with a generically
chosen hyperplane.

Let $h = \zeta_0 + \zeta_1\, x_1 + \cdots + \zeta_n \, x_n +
\zeta_{n+1} t$ be a linear form defining such a hyperplane (here, we
take $\zeta_i\in \KK$). Using it allows us to rewrite $t$
as $$\wp(x_1, \dots, x_n) = -(\zeta_0 + \zeta_1\, x_1 + \cdots +
\zeta_{n}\, x_n)/\zeta_{n+1}.$$ The isolated points in $V(\bB) \cap
V(h)$ are in one-to-one correspondence with the isolated solutions of
the system $\bB' = (b_1', \dots, b_s', b_{s+1}', \dots, b_m')$, where
$b_i' = (1-\wp)r_i + \wp g_i$, for $i=1, \dots, s$, and $(b_{s+1}',
\dots, b_m')$ are the $p$-minors of the matrix ${\bf V}' = [v_{i, j}']
= (1-\wp)\M + \wp \F \in \KK[x_1, \dots, x_n]^{p \times q}$. Hence it
is sufficient to bound the number of isolated solutions of $V(\bB')$.

For $i=1,\dots, p$ and $j=1,\dots,q$, let $\calB_{i, j}'$ be the
support of $v_{i, j}'$. We then define $\calB_j' = \cup_{1 \leq i \leq
  p} \calB_{i, j}'$, to which we add the origin if needed, and let
$\calD_{j}'$ be its Newton polytope. Similarly, for $i=1,\dots,s$ we
let $\calC_i'$ denote the Newton polytope of the support of $b_i'$. 
Then, the discussion on the number of solutions of the target 
system still applies, and shows that the system $\bB'$ admits 
at most 
\begin{equation}\label{eqdef:varrho}
\varrho = 
\sum_{\{i_1, \dots, i_{n-s}\} \subset \{1, \dots,
  q\}} {\sf MV}(\calC_1',
\dots, \calC_s', \calD_{i_1}', \dots, \calD_{i_{n-s}}')  
\end{equation}
solutions.

\paragraph{Completing the cost analysis.}
The previous discussion allows us to use the {\sf Homotopy} algorithm
from~Proposition~\ref{pro:start_mat}. In addition to the polynomials
$\g$ and matrix $\F$, we also need the combinatorial information
$\bm\omega_\i,M_\i$ described previously. The sum of the costs of
solving the start system, and of the {\sf Homotopy} algorithm is as
follow.

\begin{theorem}\label{thm:sparsedeterminantal}
  The set $V_p(\F, \g)$ admits at most $\chi$ isolated solutions,
  counted with multiplicities. There exists a randomized algorithm
  which takes $\g$, $\F$, all lifting functions $\bm\omega_i$ and
  subdivisions $\M_\i$ as input and computes a zero-dimensional
  parametrization of these isolated solutions using
  \[\softO\left(   n^5\left(
 \gamma \log(d) \chi^2 \mu w
+
\chi(\varrho + \chi^5) \binom{q}{p} \right )\right
)\]
operations in $\KK$, where
$\gamma,\chi,\varrho$ are as in
respectively~\eqref{eqdef:gamma},~\eqref{eq:zzz}
and~\eqref{eqdef:varrho}, and $\mu$ and $w$ as in~\eqref{eq:wm}.
\end{theorem}

 
\subsection{Weighted-degree polynomials}

Weighted polynomial domains are multivariate polynomial rings
$\KK[x_1, \dots, x_n]$ where each variable $x_i$ has an integer weight
$w_i \geq 1$ (denoted by $\wdeg(x_i) = w_i$). The weighted degree of a
monomial $x_1^{\alpha_1} \cdots x_n^{\alpha_n}$ is then $\sum_{i=1}^n
w_i\alpha_i$, and the weighted degree $\wdeg(f)$ of a polynomial $f$
is the maximum of the weighted degrees of its terms with non-zero
coefficients.

Weighted domains arise naturally in determining isolated critical
points of a symmetric function $\phi$ defined over a variety $V(f_1,
\ldots, f_s)$ defined by symmetric functions $f_i$.  In~\cite{FLSSV},
with J.-C. Faug\`ere, we show that the orbits of these critical points can
be described by domains of the form $\KK[e_{1, 1}, \ldots,
  e_{1,\ell_1},e_{2,1}, \ldots, e_{2,\ell_2}, \dots, e_{r, 1}, \ldots,
  e_{r,\ell_r}]$ with $e_{i, k}$ the $k$-th elementary symmetric
function on $\ell_i$ letters.  Measured in terms of these letters,
each $e_{i, k}$ has naturally weighted degree $k$.

Polynomials in weighted domains have a natural sparse structure when
compared to polynomials in classical domains.  For example, a polynomial
$p \in \KK[x_1, x_2, x_3]$ having total degree bounded by $10$ has
$286$ possible terms in a classical domain. However in a weighted
domain with weights $\bw=(5,3, 2)$ there are only $19$ possible
terms. Such a reduction also exists when considering bounds for
solutions of polynomial systems when comparing classical to weighted
domains. For instance, B\'ezout's theorem bounds the number of
isolated solutions to polynomial systems of equations by the product
of their degrees. With polynomial systems lying in a weighted
polynomial domain $\KK[x_1, \dots, x_n]$ having weights $\bw = (w_1,
\dots, w_n) \in \Z^n_{> 0}$, the weighted B\'ezout theorem (see
e.g.~\cite{James99}) states that the number of isolated points of
$V(f_1, \dots, f_n) \subset \KKbar^{n}$ is bounded by
\begin{equation}\label{weighted:bezout}
\quad \quad \quad \quad \delta = \frac{d_1 \cdots d_n}{w_1 \cdots w_n}
~~\mbox{ with } ~~ d_i = \wdeg(f_i). 
\end{equation}

In this section we show how our sparse homotopy algorithm also allows
us to describe the isolated points of $V_p(\F,\g)$ where $\F = [f_{i,
    j}] \in \KK[x_1, \dots, x_n]^{p\times q}$ and $\g = (g_1, \dots,
g_s) \in \KK[x_1, \dots, x_n]^s$ with $n=q-p+s+1$, assuming bounds on
the weighted degrees of all polynomials $f_{i,j}$ and $g_j$.  Without
loss of generality, we will assume that $w_1 \leq \cdots \leq w_n$,
and we will let $(\gamma_1, \ldots, \gamma_s)$ be the weighted degrees
of $(g_1, \ldots, g_s)$ and $(\delta_1, \ldots, \delta_q)$ be the
weighted column degrees of $\F$.

In particular, the monomial supports $\calA_1,\dots,\calA_s$ of
$g_1,\dots,g_s$ are contained in the sets $\calA'_1,\dots,\calA'_s$,
where $\calA'_i$ is the set of all $(e_1,\dots,e_n) \in \N^n$ such that
$w_1 e_1 + \cdots + w_n e_n \le \gamma_i$. Similarly, for $1\leq j
\leq q$, $\calB_j \subset \N^n$ is contained in the set $\calB'_j$ of
all $(e_1,\dots,e_n) \in \N^n$ for which $w_1 e_1 + \cdots + w_n e_n
\le \delta_j$. The sets $\calA'_i$, respectively\ $\calB'_j$, are the
supports of generic polynomials of weighted degrees at most
$\gamma_i$, respectively $\delta_j$. We denote their cardinalities by 
$a'_1,\dots,a'_s$ and $b'_1,\dots,b'_q$.

\paragraph{Representing the input.} We follow the same approach as in
the last subsection to obtain a straight-line program for
$\bB=\Theta_\rho(\frkB)$, simply by computing all monomials of
respective weighted degrees at most $(\gamma_1, \ldots, \gamma_s)$ and
$(\delta_1, \ldots, \delta_q)$, combining them to form the polynomials
$(1-t) \cdot \Theta_\rho(\frkr) + t \cdot \g$ and the matrix
$(1-t)\cdot \Theta_\rho(\frkM)+t\cdot \F$ and taking the $p$-minors of
the latter. We benefit from a minor improvement here, as for a fixed
$\gamma_i$ or $\delta_j$ we can compute all these monomials in an
incremental manner, starting from the monomial $1$, foregoing the use
of repeated squaring: this saves a factor $n \log(d)$.  Altogether,
this results in a straight-line program of size
$$\Gamma \in O \left ( (a'_1 + \cdots + a'_s + p(b'_1 + \cdots + b'_q)) +
p^4 {q \choose p} \right )$$ to compute all entries of $\bB$.

Recall that a term such as $a'_i$ denotes the number of monomials of
weighted degree at most $\gamma_i$ in $n$ variables, with $\gamma_i
\le d$ for all $i$ (and similarly for $b'_j$, for the weighted degree
bound $\delta_j$). A crude bound is thus $a'_i,b'_j \le \binom{n+d}{n}$,
resulting in the estimate
\begin{equation}\label{eqdef:beta}
\Gamma \in O \left (n^2 \binom{n+d}{n} + n^4  {q \choose p} \right ).
\end{equation}
This is not the sharpest possible bound. 
Bounding $a'_i$ by the volume of the non-negative simplex defined by
$$w_1 (e_1-1) + \cdots + w_n(e_n-1) \le \gamma_i$$ results in the
upper bound $a'_i \le (\gamma_i+w_1 + \cdots +w_n)^n/(n!w_1\cdots
w_n)$.  Using~\cite{Aharon72} and \cite[Theorem~1.1]{YauZhang06} gives
more refined results for $a'_i$ and $b'_j$ and hence also for
$\Gamma$.

\paragraph{Number of solutions of the start system.}
As in the case of sparse polynomials, we take $\rho$ in the open set
$\Open \subset \KKbar{}^N$ of Proposition~\ref{prop:hyp}
and set $\bB=\Theta_\rho(\frkB)$. In this case, the solutions of the start
system $\bB_{t=0}$ are the disjoint union of the solutions of systems
$\Theta_\rho(\frkM_{\j},\r)$, with $\frkM_{\j} = (\frkm_{i_1}, \dots,
\frkm_{i_{n-s}})$ for $\j=\{i_1, \dots, i_{n-s}\} \subset \{1, \dots,
q\}$.

By the weighted B\'ezout theorem, the system
$\Theta_\rho(\frkM_{\j},\r)$ has 
$$c_\j = \frac{\gamma_1 \cdots \gamma_s\delta_{i_1} \cdots
  \delta_{i_{n-s}}}{w_1 \cdots w_n}$$ solutions in
$\KKbar{}^{n}$. Taking the sum over all subsets $\j$ of $\{1, \ldots,
q\}$ of cardinality $n-s$, we deduce that the number of solutions of
$\bB_{t=0}$ is at most
\begin{equation}\label{weighted:c}
    c =\sum_{\i} c_\i = \frac{\gamma_1 \cdots \gamma_s  \, \eta_{n-s}(\delta_1,
      \dots, \delta_q) }{w_1 \cdots w_n},
\end{equation}
where $\eta_{n-s}(\delta_1, \dots, \delta_q)$ is the elementary
symmetric polynomial of degree $n-s$ in $\delta_1, \dots, \delta_q$.
The discussion following Proposition~\ref{prop:hyp} implies that the
system $\bB_{t=1}$ which we want to solve admits at most $c$
isolated solutions.

\paragraph{Solving the start system.}
To find these solutions, as in the previous subsection, we solve all
systems $\Theta_\rho(\frkM_{\j},\r)$ independently. We are not aware
of a dedicated algorithm for weighted-degree polynomial systems whose
complexity would be suitable; instead, we rely on the geometric
resolution algorithm as presented in~\cite{GiLeSa01}. In what follows,
our first requirement is that $\rho$ be in the open set $\Open \subset
\KKbar{}^N$ of Proposition~\ref{prop:hyp}, but  we will add finitely
many Zariski-open conditions on~$\rho$.
 
For a subset $\i = \{i_1, \ldots, i_{n-s}\}\subset \{1, \ldots, q\}$,
let $(d_{\i,1},\dots,d_{\i,n})$ denote the sequence $(\gamma_1,
\ldots, \gamma_s, \delta_{i_1},\dots, \delta_{i_{n-s}})$ sorted in
non-decreasing order; we write
\begin{equation} \label{eq:kappa}
  \kappa_\i = \max_{1\leq k \leq
    n}(d_{\i,1} \cdots d_{\i,k} w_{k+1} \cdots w_n)\quad \text{ and
  }\quad \kappa = \sum_{\i = \{i_1, \ldots, i_{n-s}\}\subset \{1,
    \ldots, q\}} \kappa_\i.
\end{equation}
Recall as well that we set $d =
\max(\gamma_1,\dots,\gamma_s,\delta_1,\dots,\delta_q)$. 

\begin{lemma}\label{prop:degreebound} 
   For $\i =\{i_1, \dots, i_{n-s}\} \subset \{1, \dots, q\}$, and a
   generic $\rho \in \KKbar{}^N$, one can solve
   $\Theta_\rho(\frkM_{\j},\r)$ by a randomized algorithm that uses
   \[
   \softO \left (n^4 \Gamma d^2 \left (\frac{\kappa_i}{w_1\cdots
         w_n}\right )^2\right ) 
   \]
   operations in $\KK$.
\end{lemma}
\begin{proof}
  The polynomials $\Theta_\rho(\frkM_{\j},\r)$ have weighted degrees
  at most $(\gamma_1, \ldots, \gamma_s, \delta_{i_1}, \ldots,
  \delta_{i_{n-s}})$. We first reorder these equations
 in non-decreasing order of weigthed
  degree; we write the reordered sequence of polynomials as
  $(h_1,\dots,h_n)$, their respective weighted degrees being at most
  $(d_{\i,1},\dots,d_{\i,n})$.

  By Proposition~\ref{prop:start_pro}, since the supports of
  $\frkM_{\j}$ and $\r$ contain the origin, for a generic choice of
  $\rho$, the equations $\Theta_\rho(\frkM_{\j},\r)$ define a reduced
  regular sequence (possibly terminating early and thus defining the
  empty set). We can thus apply the geometric resolution algorithm as
  in~\cite[Theorem~1]{GiLeSa01}.   

  The algorithm in~\cite{GiLeSa01} takes its input
  represented as a straight-line program. To obtain one, we take our
  straight-line program of length $\Gamma$ that computes $\bB$ and set
  $t=0$; the resulting straight-line program computes all
  $\Theta_\rho(\r)$ and $\Theta_\rho(\frkm_1,\dots,\frkm_q)$, and in
  particular $\Theta_\rho(\frkM_{\j})$.
  We deduce that we can compute a zero-dimensional parametrization of
  the solutions of $\Theta_\rho(\frkM_{\j},\r)$ using $ \softO (n^4
  \Gamma d^2 \mathbf{\Sigma}_\i^2 )$ operations in $\KK$. Here,
  $\mathbf{\Sigma}_\i$ is the maximum of the degrees of the
  ``intermediate varieties'' $V_1,\dots,V_n$, where $V_i$ is defined
  by the first $i$ equations in $\Theta_\rho(\frkM_{\j},\r)$.  Hence,
  to conclude, it suffices to prove that $\mathbf{\Sigma}_\i \le
  \kappa_\i/(w_1\cdots w_n)$.

  Fix an index $\ell$ in $\{1,\dots,n\}$. We identify degree-1
  polynomials $ P = p_{ 0} + p_{ 1} x_1 + \cdots + p_{ n} x_n$ in
  $\KKbar[x_1, \ldots, x_n]$ with points in $\KKbar{}^{n+1}$. Then,
  there exists a non-empty Zariski open set $\mathscr{P}\subset
  \KKbar{}^{(n+1)(n-\ell)}$ such that for $(p_{i,j})_{0\leq j\leq n,
    1\leq i \leq n-\ell}\in \mathscr{P}$, defining $P_i$ as
  \[
    P_i = p_{i, 0} + p_{i, 1} x_1 + \cdots + p_{i, n} x_n
  \] 
  implies that $V_\ell \cap V(P_1) \cdots \cap V(P_{n-\ell})$ has
  cardinality $\deg(V_\ell)$. Up to taking the $p_{i,j}$'s in the
  intersection of $\mathscr{P}$ with another non-empty Zariski open
  set, one can perform Gaussian elimination to rewrite
  $P_1,\dots,P_{n-\ell}$ as
  \[
  x_{\ell+1} - \wp_{\ell+1}(x_1, \dots, x_{\ell}),
  \dots,
  x_{n} - \wp_n(x_1, \dots, x_{\ell}).
  \]
  For $k=1,\dots, \ell$, let $ g_k(x_1, \dots, x_\ell) = h_k(x_1,
  \dots, x_\ell, \wp_{\ell+1}(x_1, \dots, x_\ell), \dots, \wp_n(x_1,
  \dots, x_{\ell}))$ in $\KK[x_1, \dots, x_\ell]$. Because the
  sequence of weights is non-decreasing, these have respective
  weighted degrees at most $d_{\i,1},\dots,d_{\i,\ell}$ and, by construction,
  $V(g_1, \dots, g_\ell)$ is finite and $\deg(V_\ell) = \deg(V(g_1,
  \dots, g_\ell))$. Using the weighted B\'ezout's theorem implies
  \[
    \deg(V(g_1, \dots, g_\ell)) \leq 
    \frac{d_{\i,1} \cdots d_{\i,\ell}}{w_1 \cdots w_\ell} = 
\frac{ d_{\i,1} \cdots d_{\i,\ell} w_{\ell+1} \cdots w_n}{w_1 \cdots w_n}
=\frac {\kappa_\i}{w_1 \cdots w_n}.
\qedhere
\]
Taking all possible $\i$ into account, we see that for a generic
$\rho$ we can compute zero-dimensional parametrizations for all
$\Theta_\rho(\frkM_{\j},\r)$ using $$ \softO \left (n^4 \Gamma d^2
\left (\frac{\kappa}{w_1\cdots w_n}\right )^2\right )$$ operations in
$\KK$. As in the previous subsection, taking the union of all these
parametrizations does not introduce any added cost.
\end{proof}

\paragraph{Degree of the homotopy curve.}
Finally, we need an upper bound on the precision $t^e$ to which
we do the computations. As before, a suitable upper bound is the
number of isolated intersection points in $\KKbar{}^{n+1}$ between
$V(\bB)$ and a generic hyperplane.

Let $ \zeta=\zeta_0 + \zeta_1\, x_1 + \cdots + \zeta_n \, x_n +
\zeta_{n+1} t$ be a linear form defining such a hyperplane (here, we
take $\zeta_i\in \KK$). We are interested in counting the isolated
solutions of all equations $\g'=(\zeta,(1-t) \cdot \Theta_\rho(\frkr) + t
\cdot \g)$, and all $p$-minors of $\F'= (1-t)\cdot \Theta_\rho(\frkM) + t
\cdot \F$, that is, of $V_p(\F',\g')$.

Assign weight $w_t=1$ to $t$, so the weighted degree of $\zeta$ is
$w_n$. Then, the system above is of the kind 
considered in this section, but with $n+1$ variables instead of $n$,
and $s+1$ equations $\g'$ instead of $s$.  The weighted degrees of the
equations $\g'$ are $(w_n,\gamma_1+1,\dots,\gamma_s+1)$ and the weighted
column degrees of $\F'$ are $(\delta_1+1,\dots,\delta_q+1)$. As we
pointed out when counting the solutions of the start system, this
implies that our equations admit at most $e$ isolated solutions,
with
\begin{equation}\label{weighted:e}
    e  = \frac{(\gamma_1+1) \cdots (\gamma_s+1)  \, \eta_{n-s}(\delta_1+1,
      \dots, \delta_q+1) }{w_1 \cdots w_{n-1}},
\end{equation}
where $\eta_{n-s}$ is the elementary symmetric polynomial of degree
$n-s$.

\paragraph{Completing the weighted homotopy algorithm.}
The previous paragraphs allow us to use the {\sf Homotopy} algorithm
from~Proposition~\ref{pro:start_mat}; we obtain the following result.

\begin{theorem}\label{thm:homotopy}
  The set $V_p(\F, \g)$ admits at most $c$ isolated solutions, counted
  with multiplicities. There exists a randomized algorithm which
  takes $\g$ and $\F$ as input and computes a zero-dimensional
  parametrization of these isolated solutions 
  using
  \[\softO\Big( \big(c(e + c^5)+  \, d^2 \,
  \big(\frac{\kappa}{ w_1\cdots w_n}\big)^2\big)n^4 \Gamma\Big)\]
  operations in $\KK$, where
  $\Gamma,c,\kappa,e$ are as in
  respectively~\eqref{eqdef:beta},~\eqref{weighted:c},~\eqref{eq:kappa}
  and~\eqref{weighted:e}.
\end{theorem}

\section{Example}
\label{sec:ex}
In this section  we provide an example illustrating the steps of our homotopy algorithm. 
Let  
$$\g = ( 99x_1^3+92x_1^2-228x_1x_2+67x_1-140x_2+98x_3+25)  \in \Q[x_1, x_2,
x_3]$$ and $\F \in \Q[x_1, x_2, x_3]^{2 \times 3} $ be
{\scriptsize
\[
\left( \begin{array}{ccc} 9x_1^2+65471x_1+59x_2+42308x_3+65504 &
86x_1^2+65460x_1+65414x_2+12381x_3+44  &   65477x_1+59898x_3+76 \\ & & \\
 65501x_1^2+51x_1+65466x_2+57496x_3+35 & 16x_1^2+99x_1+65503x_2+17950x_3+31 & 65454x_1+
41178x_3 + 65453
\end{array} \right) .
\]
}
The support of $g$ is $\calA = \{(3, 0, 0), (2, 0, 0), (1, 1, 0), (1,
0, 0), (0, 1, 0), $ $(0, 0, 1), (0, 0, 0)\} \subset \Z^3$ with
unions of the column supports of $\F$ being
\begin{align*}
\calB_1 &= \{(2, 0, 0), (1, 0, 0), (0, 1, 0), (0, 0, 1), (0, 0, 0)\},
  \\ \calB_2 &=
\{(2, 0, 0), (1, 0, 0), (0, 1, 0), (0, 0, 1), (0, 0, 0)\},  \\
   \calB_3 &= \{(1, 0, 0),
(0, 0, 1), (0, 0, 0)\}. 
\end{align*}
\paragraph*{\bf{Start system.}}
The start system for $(\F, g)$ is built as follows. Let 
$r = ~88x_1^3~-~82x_1^2~-~70x_1x_2\\~+~41x_1~+~91x_2~+~29x_3+70
\in \Q[x_1, x_2, x_3]$ a polynomial supported by
$\calA$ and define
 $m_1 =~-78x_1^2~-~4x_1~+~5x_2~-~91x_3~ -~ 44,~ m_2 =
 63x_1^2+10x_1-61x_2-26x_3-20,\, {\rm and} \, m_3 = 88x_1+95x_3 + 9$,  
polynomials
 in $\Q[x_1, x_2, x_3]$ supported by $(\calB_1, \calB_2,
 \calB_3)$.  The starting polynomial system $\r = (r)$ and the
 start matrix are given as
\[
\M = \begin{pmatrix}
-62m_1 & 26 m_2 & 10m_3 \\
-83 m_1 & -3 m_2 & -44m_3
\end{pmatrix} \in \Q[x_1, x_2, x_3]^{2 \times 3}. 
\]
We remark that the coefficients in the start vector and start matrix for this example were chosen randomly, in this case with the help of the rand() command in Maple.

\paragraph*{\bf{A parametrization of the start system.}}
The set of $2$-minors of $\M$ is given by $(2344 m_1 m_2,\\ 3558 m_1 m_3,
-1114m_2m_3)$ and hence  $V_2(\M, r) = V_1 \cup V_2 \cup V_3$, where
\[
V_1 =  V(m_1, m_2, r), \, V_2 = V(m_1, m_3, r), \, {\rm and} \,V_3 =
V(m_2, m_3, r). 
\] 
Parametrizations of $V_1, V_2,$ and $V_3$ are given by 
{\footnotesize
\begin{align*}
\scrR_{0, 1} = ~~ &
\big((10671923044484y^3+164650405712264y^2+ 
 541980679674061y  
~+  393540496795784,  \\ &
  \frac{23707677043321206}{205138445880446701}y^2 + 
\frac{197994419338092137}{205138445880446701} y  
 + \frac{3859258707817950}{205138445880446701},  \\ & 
\frac{2817387683743776}{205138445880446701}y^2 
~-  \frac{334804957251324375}{205138445880446701}y-  
\frac{199554818581221524}{205138445880446701},     
  y), x_3\big), \\
\scrR_{0, 2} = ~~& \big(
                             (1076005625y^3+2749690925y^2+2278375403y+797867887,
  \\
&~~  -\frac{95}{88}y -\frac{9}{88}, \frac{70395}{3872}y^2 + \frac{201161}{9680}y +
  \frac{171943}{19360}, y), x_3\big), \\
\scrR_{0, 3}= ~~& 
\big(
  (410682625y^3+773879025y^2+2045246267y-666910765, \\
&~~-\frac{95}{88}y -\frac{9}{88},
  \frac{568575}{472384}y^2-\frac{88607}{236192}y-\frac{157697}{472384}, y), x_3\big).
\end{align*}
}
Taking the union of $(\scrR_{0, i})_{1 \leq i \leq 3}$ gives a parametrization
$\scrR_0$ of $V_p(\M, r)$ with 
{\footnotesize
\begin{align*}\scrR_0 = ~~&   ((q_0, v_{0, 1}, v_{0, 2}, v_{0, 3}), \Lambda_0)\\
= ~~&  \big((4715888798904593238258009062500y^9+ \cdots, \\
&~~ \frac{10476346966766553878790167132343750}{205138445880446701}y^8 + 
\cdots, \\
&~~ \frac{2265193491697540283699777221137124035318470625}{24226029904697233601296}y^8+
\cdots, \\
&~~15866264491953179878625y^7+ \cdots
 ), x_3\big).
\end{align*}
}

\paragraph*{\bf{Degree bounds.}} 
The mixed volumes associated to our
square sub-systems are ${\sf MV}_1 = {\sf MV}(\conv(\calA), \conv(\calB_1),
\conv(\calB_2)) = 3$, ${\sf MV}_2 = {\sf MV}(\conv(\calA), \conv(\calB_1),
\conv(\calB_3)) = 3$, and finally ${\sf MV}_3 = {\sf MV}(\conv(\calA),$ $\conv(\calB_2),
\conv(\calB_3)) = 3$. So $\chi = {\sf MV}_1+{\sf MV}_2+{\sf MV}_3 = 9$ which
is a bound on the number of isolated solutions of $V_2(\F, g)$. Note
that this number coincides with  the actual number of 
isolated solutions of $V_2(\M, r)$ as the degree of $q_{0}$ equals $9$. 

\paragraph*{\bf{A parametrization $\scrR_1$ of $V_2(\F, g)$.}} 
We apply the
{\sf Homotopy} algorithm to the system $(M_2((1-t)\F + t\M), (1-t)r +
tg)$ and $\scrR_0$ to obtain $\scrR_1$.  As the coefficients of the 
result over $\Q$ are quite large we  illustrate this  calculation over
$\mathbb{F}_{65521}$, the finite field of 
$65521$ elements.  In this case we obtain
%
{\footnotesize
\begin{align*}\scrR_0 = ~~& \big(
                        (y^9+42377y^8+63439y^7+23268y^6+1541y^5+21916y^4
 \\
&~+24479y^3+ 1064y^2+47617y+765, 18447y^8+58286y
^7+48619y^6  
 \\ 
&~+49312y^5+~42721y^4+  44021y^3+47621y^2+39038y+13072, \\ 
&~ 9852y^8   
+30892y^7+~29236y^6+ 63043y^5+623y^4+8249y^3 \\
&~+ 22956y^2+ 23577y + 41427,~3y^7+19233y^6+56323y^5+58151y^4 \\
&~ +  8939y^3 + 30577y^2+ 13156y), x_3\big) 
\end{align*}
}
and
{\footnotesize
\begin{align*}
\scrR_1 =~~&  \big( (y^9+27502y^8+1022y^7+42474y^6+21370y^5+47501y^4 \\
&~+37694y^3+~ 13474y^2+49870y+26489,
  19690y^8+28497y^7 \\
&~+23045y^6+29265y^5+~ 32212y^4+8948y^3+16460y^2 \\ 
&~+19357y +9600,
  26426y^8+24119y^7+~ 48429y^6+34031y^5\\ 
&~+32994y^4 +13559y^3+34993y^2+59636y+64778, y),
  x_3\big). 
\end{align*}
}

We note that using the non-sparse homotopy algorithm 
from \cite{HSSV18} produces a degree bound of $24$, a considerable
over estimate of the number of isolated zeros.

\section{Topics for future research}

We have presented a new homotopy algorithm for determining isolated
solutions of algebraic sets $V_p(\F,\g)$ for $\F$ a $p \times q$
matrix and $\g$ a vector having entries from a multivariate polynomial
domain. Our algorithm determines the bounds central to homotopy
algorithms based on the column support of the matrix $\F$.  Our column
supported homotopy algorithm can be applied to the case where our
entries come from a weighted polynomial domain.  Such weighted domains
arise when we determine the isolated critical points of a symmetric
function $\phi$ defined over a variety $V(\f)$ generated by symmetric
functions in $\f$. The resulting complexity is improved by a factor
depending on the size of the symmetric group.

Still regarding critical point computations, but for non symmetic
input $\F,\g$, the natural bounds for a sparse homotopy would come from
considering the row support rather than the column support of $\F$. An
interesting approach would be the follow the algorithm given
in~\cite{HSSV18} for dense polynomials. However, proving that in the
sparse case, the corresponding start systems satisfy the genericity
properties we need is not straightforward; this is the subject of
future work.

\paragraph*{Acknowledgements.} G. Labahn is supported by the 
Natural Sciences and Engineering Research Council of Canada (NSERC), 
grant number RGPIN-2020-04276.  \'E. Schost is supported
by an NSERC Discovery Grant.  T.X. Vu is supported by a labex
CalsimLab fellowship/scholarship. The labex CalsimLab, reference
ANR-11-LABX-0037-01, is funded by the program ``Investissements
d'avenir'' of the Agence Nationale de la Recherche, reference
ANR-11-IDEX-0004-02. M. Safey El Din and T.X. Vu are supported by the ANR grants
ANR-18-CE33-0011 \textsc{Sesame}, ANR-19-CE40-0018 \textsc{De Rerum Natura} and
ANR-19-CE48-0015 \textsc{ECARP}, the PGMO grant \textsc{CAMiSAdo} and the
European Union's Horizon 2020 research and innovation programme under the Marie
Sklodowska-Curie grant agreement N. 813211 (\textsc{POEMA}).



\end{document}